\pdfoutput=1 

\documentclass[11pt]{article}
\usepackage[margin=1in,footskip=0.25in]{geometry}

\newif\ifdraft
\drafttrue

\usepackage{hyperref}

\usepackage{amssymb}
\usepackage{amsfonts}
\usepackage{amsmath}
\usepackage{amsthm}
\usepackage{dsfont}
\usepackage{mathtools}
\usepackage{bbm}
\usepackage{cite}

\usepackage[labelsep=space]{caption}
\usepackage{subcaption}

\usepackage[textsize=tiny,disable]{todonotes}

\newcommand{\htodo}[1]{\todo[color=orange]{#1}}

\usepackage{comment}

\newcommand{\IN}{\mathds{N}}
\newcommand{\EE}{\mathcal{E}}
\renewcommand{\AA}{\mathcal{A}}
\newcommand{\BB}{\mathcal{B}}

\newcommand{\HH}{\mathcal{H}}

\newcommand{\XX}{\mathcal{X}}

\DeclareMathOperator{\Bin}{Bin}
\newcommand{\Exp}[1]{\mathbb{E}\left[\,#1\,\right]}

\usepackage{setspace}
\usepackage{xspace}
\usepackage{xfrac}

\usepackage{tikz}
\usepackage{epsfig}
\usepackage{float}
\usepackage{placeins}

\usepackage{multirow}

\usepackage{enumerate}

\renewcommand{\Pr}[1]{\mathbb{P}\left[#1\right]}

\newcommand\E[1]{\mathbb{E}\left[#1\right]}

\newcommand{\expp}[1]{\exp\left( #1 \right)}

\DeclarePairedDelimiter{\ceil}{\lceil}{\rceil}
\DeclarePairedDelimiter{\floor}{\lfloor}{\rfloor}

\usepackage{xcolor} 
\usepackage{booktabs} 

\newtheorem{theorem}{Theorem}
\newtheorem{lemma}[theorem]{Lemma}

\newtheorem{claim}[theorem]{Claim}

\newcommand{\push}{\textsc{push}\xspace}

\newcommand{\ppull}{\textsc{push-pull}\xspace}
\newcommand{\visit}{\textsc{visit-exchange}\xspace}
\newcommand{\tvisit}{\textsc{t-visit-exchange}\xspace}
\newcommand{\rvisit}{\textsc{r-visit-exchange}\xspace}
\newcommand{\meet}{\textsc{meet-exchange}\xspace}
\newcommand{\tmeet}{\textsc{t-meet-exchange}\xspace}

\newcommand{\Tpush}{T_{\rm push}}
\newcommand{\Tppull}{T_{\rm ppull}}
\newcommand{\Tvisit}{T_{\rm visitx}}
\newcommand{\Tmeet}{T_{\rm meetx}}
\newcommand{\Rvisit}{R_{\rm visitx}}

\title{How to Spread a Rumor: Call Your Neighbors or Take a Walk?}

\author{
    George Giakkoupis\\ 
    INRIA, Rennes, France
    \and
    Frederik Mallmann-Trenn \\
    King's College London, UK
    \and
    Hayk Saribekyan\\ 
    University of Cambridge, UK
}

\date{}

\begin{document}

\maketitle

\begin{abstract}
We study the problem of randomized information dissemination in networks.
We compare the now standard \ppull protocol, with agent-based alternatives where information is disseminated by a collection of agents performing independent random walks.
In the \visit protocol, both nodes and agents store information, and each time an agent visits a node, the two exchange all the information they have.
In the \meet protocol, only the agents store information, and exchange their information with each agent they meet.

We consider the broadcast time of a single piece of information in an $n$-node graph for the above three protocols,
assuming a linear number of agents that start from the stationary distribution.
We observe that there are graphs on which the agent-based protocols are significantly faster than \ppull, and graphs where the converse is true.
We attribute the good performance of agent-based algorithms to their inherently fair bandwidth utilization, and conclude that, in certain settings, agent-based information dissemination, separately or in combination with \ppull, can significantly improve the broadcast time.

The graphs considered above are highly non-regular.
Our main technical result is that on any regular graph of at least logarithmic degree, \ppull and \visit have the same asymptotic broadcast time.
The proof uses a novel coupling argument which relates the random choices of vertices in \ppull with the random walks in \visit.
Further, we show that the broadcast time of \meet is asymptotically at least as large as the other two's on all regular graphs, and strictly larger on some regular graphs.

As far as we know, this is the first systematic and thorough comparison of the running times of these very natural information dissemination protocols.
\end{abstract}

\section{Introduction}
\label{sec:introduction}

We investigate the problem of spreading information (or rumors) in a distributed network using randomized communication.
The archetypal paradigm solution is the so-called, \emph{randomized rumor spreading} protocol, where each informed node samples a random neighbor in each round, and sends the information to it.
This is the \push version of rumor spreading, introduced by Demers et al.\ in the 80's~\cite{Demers1988}, as a robust and lightweight protocol for distributed maintenance of replicated databases~\cite{Demers1988,Feige1990}.

The \ppull variant of rumor spreading, popularized by Karp et al.\ in 2000~\cite{KarpSSV00}, allows for bidirectional communication: In each round, every node calls a random neighbor and the two nodes exchange all information they have.
\ppull was initially proposed as a way to reduce the message complexity of \push on the complete graph~\cite{KarpSSV00}.
It was subsequently observed that it is significantly faster than \push in several families of graphs, including graph models of social networks~\cite{Chierichetti2011tcs,DoerrFF11}.

The above two protocols have been studied extensively over the past 15 years, and have also found several applications, including data aggregation~\cite{Kempe2003,Boyd2006,Mosk-Aoyama2008}, resource discovery~\cite{Balter1999}, failure detection~\cite{vanRenesse1998}, and even efficient simulation of arbitrary distributed computations~\cite{Censor-Hillel2017}.

We compare the above well-established protocols for information spreading, with agent-based alternatives that have received almost no attention so far, even though they have very attractive properties, as we will see.
These alternative protocols use a collection of agents performing independent random walks to disseminate information.
In the \visit protocol, both nodes and agents store information, and each time an agent visits a node, the two exchange all the information they have.
In the \meet protocol, only the agents store information, and exchange their information with each agent they meet.

Independent parallel random walks have been studied since the late 70s~\cite{Aleliunas1979}, mainly as a way to speed-up cover and hitting times and related graph problems~\cite{Broder1994,AlonAKKLT11,Elsasser2011,Efremenko2009}.
As far as we know, \visit has not been studied before.
For \meet there is some limited previous work.
It was studied for specific graph families, namely grids~\cite{PPPU10,LamLMSW12} and random graphs~\cite{CooperFR09}.
Also, general bounds on the broadcast time of \meet with respect to the meeting time were shown~\cite{DimitriouNS06}.

In this paper, we restrict our attention to the case where the number of agents in the network is linear in the number of nodes $n$, and we assume that all agents start from the stationary distribution.

Under the assumption that there is a linear number of agents, the agent-based protocols have similar amount of communication as the rumor spreading protocols, both in terms of the (maximum) total number of messages sent per round, which is linear, and the total number of bits.
One can think of the agents simply as tokens passed between nodes, along with the actual information (if there is any).
Agents need not be labeled, so each node only needs to send a counter of the number of agents in each message.

The assumption that agents start from the stationary distribution makes sense in a setting where several pieces of information (or rumors) are generated frequently and distributed in parallel over time by the same set of agents, which execute perpetual independent random walks.
As discussed later, our results for regular graphs hold also in the case where there is exactly one agent starting from each node.

One distinct advantage of the agent-based protocols is their \emph{locally fair} use of bandwidth, i.e., all edges are used with the same frequency, since the random walks are independent and start from stationarity.
Interestingly, the superiority of \ppull over \push is commonly attributed to a similar fairness property: that nodes of larger degree contribute more to the dissemination ---
except that \ppull satisfies this property only for \emph{some} graph topologies, and \emph{approximately}, as we will see below.
In the agent-based protocols, on the other hand, this property is satisfied in a very precise and exact way.

We will see that this fairness property results in a significant performance advantage of \visit and \meet over \push and \ppull in certain families of graphs, on which the first two processes need only logarithmic time to spread an information, whereas the other two need polynomial time.


\paragraph{Contribution.}

We compare the broadcast times of a single piece of information, originated at an arbitrary node $s$ of an $n$-node graph $G=(V,E)$, when \push (or \ppull), \visit, and \meet are used.
In the first three, the broadcast time is the time until all vertices are informed, while in \meet it is the time until all \emph{agents} are informed.
Also, for \meet, we assume that the first agent to visit the source $s$ becomes informed, and from that point on, information is exchanged only between agents.\footnote{This is a technicality used to allow for direct comparison between the protocols, and has limited effect on our results.}
As mentioned before, we assume a linear number of agents, each starting from the stationary distribution.

\begin{figure*}[t]
    \centering
    \begin{subfigure}{0.195\textwidth}
      \centering
    \includegraphics[width=\linewidth,page=1]{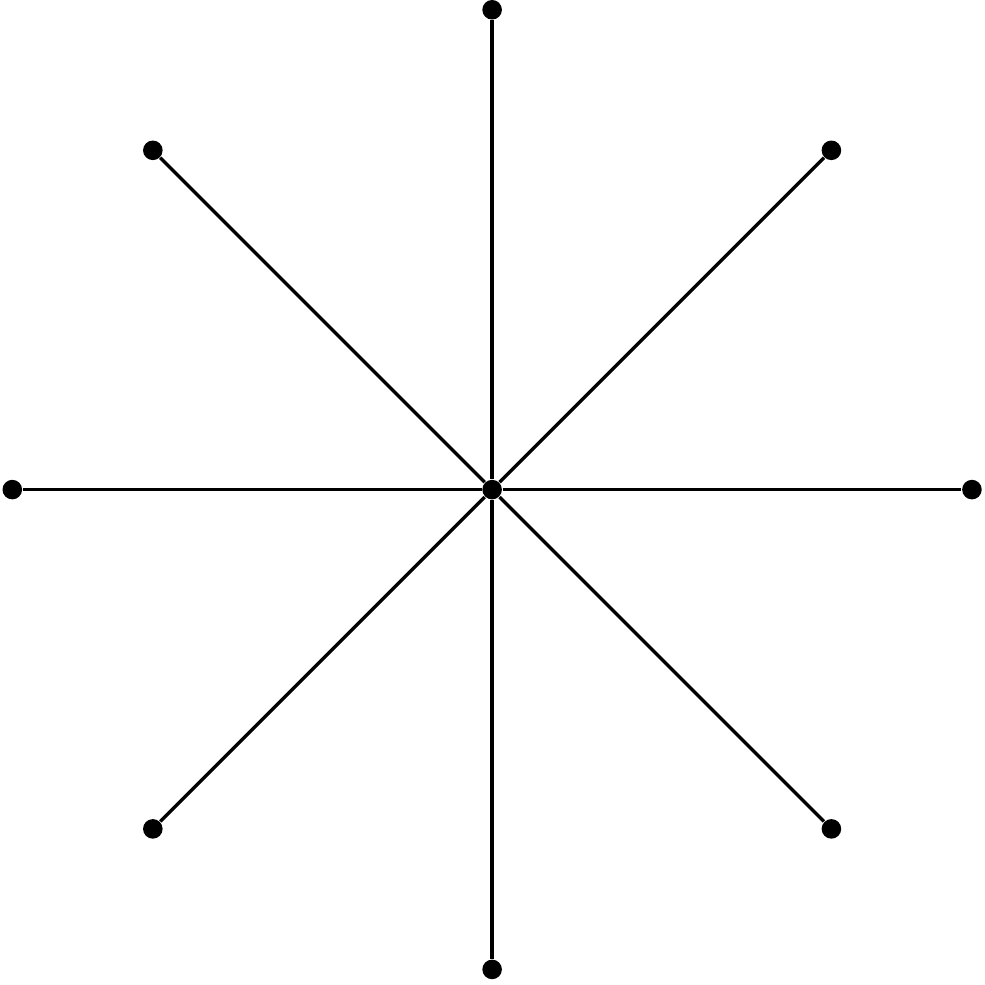}
      \caption{ }
      \label{fig:sub1}
    \end{subfigure}%
    \begin{subfigure}{0.195\textwidth}
      \centering
    \includegraphics[width=\linewidth,page=3]{pics/stars}
      \caption{}
      \label{fig:sub2}
    \end{subfigure}
    \begin{subfigure}{0.195\textwidth}
      \centering
    \includegraphics[width=\linewidth,page=3]{pics/graph_examples}
      \caption{}
      \label{fig:sub3}
    \end{subfigure}
    \begin{subfigure}{0.195\textwidth}
      \centering
    \includegraphics[width=\linewidth,page=4]{pics/graph_examples}
      \caption{}
      \label{fig:sub4}
    \end{subfigure}
    \begin{subfigure}{0.195\textwidth}
      \centering
    \includegraphics[width=\linewidth,page=5]{pics/graph_examples}
      \caption{}
      \label{fig:sub5}
    \end{subfigure}
    \caption{
        \textbf{(a)}~Star $S_n$, on which $\Exp{\Tpush} = \Omega(n \log n)$ and all other processes take $O(\log n)$ time w.h.p.
        \textbf{(b)}~Double-star $S_n^2$, on which $\Exp{\Tppull} = \Omega(n)$, and $\Tvisit,\Tmeet=O(\log n)$ w.h.p.
        \textbf{(c)}~Heavy binary tree $B_n$ (leaves are connected to a clique), on which $\Tpush=O(\log n)$ w.h.p., $\Exp{\Tvisit} = \Omega(n)$, and, for a leaf source, $\Tmeet=O(\log n)$ w.h.p.
        \textbf{(d)}~Siamese heavy binary tree $D_n$, on which $\Tpush=O(\log n)$ w.h.p., and $\Exp{\Tvisit},\Exp{\Tmeet} = \Omega(n)$.
        \textbf{(e)}~Cycle-of-stars-of-cliques ($n^{1/3}$ stars with $n^{1/3}$ leaves each, $n^{1/3}$ nodes per clique), on which  $\Exp{\Tvisit} = O(n^{2/3})$ and $\Exp{\Tmeet} = \Omega(n^{2/3}\log n)$.
    }
    \label{fig:examples}
\end{figure*}

We observe that in general graphs, the broadcast times of the above protocols are \emph{incomparable}:
For any pair of protocols, there are examples of graphs where the first protocol is significantly faster than the other, by a polynomial factor in most cases.
The examples we use, depicted in Fig.~\ref{fig:examples}, are fairly simple, mainly trees or superpositions of trees with cliques.

The star graph in Fig.~\ref{fig:examples}(a) is an example where \push is known to take $\Omega(n\log n)$ rounds, as the center must contact all leaves.
\visit and \meet, on the other hand, take only logarithmic time,
as roughly half of the walks visit the center in each round, and a constant number visits each leaf on average.


In the star, \ppull is also (extremely) fast.
The next example, the double-star in Fig.~\ref{fig:examples}(b), is a graph where \ppull (and thus also \push) is slow, whereas \visit and \meet are still fast.
This 
demonstrates the advantages of the local fairness property we pointed out earlier, and the impact it can have on the broadcast time:
Here \ppull
selects the edge between the two stars only with probability $O(1/n)$, %
which results in an expected broadcast time of $\Omega(n)$.
In \visit and \meet, on the other hand,
the probability that some agent crosses the edge in a round is constant,
resulting in a logarithmic broadcast time.

Fig.~\ref{fig:examples}(c) and Fig.~\ref{fig:examples}(d) illustrate examples where rumor spreading protocols have an advantage over agent-based protocols.
In both examples \push (and thus \ppull) has logarithmic broadcast time.
For \visit, at least linear time is needed: Since almost all the volume of the graph is concentrated on the leaves, it is likely that all agents are on the leaves at time zero, and then it takes linear time before the first walk reaches the root.
For \meet, we have that it is fast in the first example, as all walks meet quickly in the clique induced by the leaves.
However, in the second example, where agents are roughly split between the two induced cliques, the broadcast times of both \meet and \visit is $\Omega(n)$.

The above results suggest that in certain settings, agent-based information dissemination, separately or in combination with \ppull, can significantly improve the broadcast time.
We stress that, even though the examples presented may seem contrived, they are intentionally simple to demonstrate the principle reasons that make the protocols perform differently, and we expect that similar result can be observed in a wide range of networks.
In particular, we believe that the observations for the double-star example of Fig.~\ref{fig:examples}(b),  extend to more general tree-like topologies with high-degree internal nodes.

All examples we have discussed so far, involve highly non-regular graphs.
Our main technical result concerns regular graphs, and can be stated somewhat informally as follows.
(For the formal, stronger statements see Sections~\ref{sec:push<msg} and~\ref{sec:msg<push}.)

\begin{theorem}
    \label{thm:intro}
    For any $d$-regular graph on $n$ vertices, where $d = \Omega(\log n)$, and any source vertex, the broadcast times of \push and \visit are asymptotically the same both in expectation and w.h.p.,\footnote{By with high probability (w.h.p.) we mean with probability at least $1-n^{-c}$, with some constant $c>0$ that can be made arbitrary large, by adjusting the constants in the statement.}
    modulo constant multiplicative factors.
\end{theorem}

Recall that \push and \ppull have asymptotically the same broadcast times on regular graphs~\cite{GiakkoupisNW16}.
Note also that the broadcast times of \push and \ppull on $d$-regular graphs can vary from logarithmic, e.g., in random $d$-regular graphs, to polynomial, e.g., in a path of $d$-cliques where the broadcast time is $\Omega(n)$.

The proof of Theorem~\ref{thm:intro} uses a novel coupling argument which relates the random choices of vertices in \push, with the random walks in \visit.
Roughly speaking, for each node $u$, we consider the list of neighbors that $u$ samples in \push, and the list of neighbors to which informed agents move to in their next step after visiting $u$ in \visit.
Our coupling just sets the two lists to be identical for each $u$.
Even though the coupling is straightforward, its analysis is not.
On the one direction of the proof, showing that the broadcast time of \push is dominated by the broadcast time of \visit, the main step is to bound the congestion, i.e., the number of agents encountered along a path, for all possible paths through which information travels.
On the reverse direction, we focus only on the fastest path through which information reaches each node in \push, and show that an equally fast path exists in \visit.
A useful trick we devise, to consider only every other round of \visit in the coupling, simplifies the proof of this second direction.
We expect that our proof ideas will be useful in other applications of multiple random walks as well.

In addition to Theorem~\ref{thm:intro}, we observe that the broadcast time of $\meet$ is asymptotically at least as large as \visit's on any regular graph of at least logarithmic degree.
The idea is that once all agents are informed it takes at most logarithmic time to cover the graph.
It is probably surprising that the converse direction is not true, i.e., there are regular graphs where \meet is strictly slower than \visit.
Fig.~\ref{fig:examples}(e) presents one such example of a $d$-regular graph, where $d = n^{1/3}$, for which a logarithmic-factor gap exists between the broadcast times of the two protocols.

\paragraph{Road-map.}
In Section~\ref{sec:related-work}, we survey additional related work.
In Section~\ref{sec:processes}, we provide a formal description of the protocols we study.
In Section~\ref{sec:cexamples}, we analyze the broadcast times for the example graphs in Fig.~\ref{fig:examples}.
In Section~\ref{sec:push<msg}, we prove the first direction of Theorem~\ref{thm:intro}, namely, that \push is at least as fast as \visit;
the other direction is proved in Section~\ref{sec:msg<push}.
The result that \visit is at least as fast as \meet on regular graphs is provided in Section~\ref{sec:visit<meet}.
Finally, some open problems are discussed in Section~\ref{sec:open}.


\section{Related work}
\label{sec:related-work}

The \push variant of rumor spreading was first considered in~\cite{Demers1988}.
It was subsequently analyzed on various graphs in~\cite{Feige1990}, where also bounds with the degree and diameter were shown for general graphs.
The \ppull variant was introduced in \cite{KarpSSV00}, and was studied initially on the complete graph.
More recently, there has been a lot of work on showing that in several settings $O(\log n)$ rounds of rumor spreading suffice w.h.p.\ to broadcast information~\cite{doerr2011social,BerenbrinkEF16,DoerrFS14}.
In addition, general bounds in terms of expansion parameters of the graph have been studied extensively, e.g., in~\cite{Giakkoupis2014_vertex_tight,ChierichettiGLP18}.

Another line of work compares synchronous and asynchronous versions of rumor spreading, where in the latter each node takes steps at the arrival times of an independent unit-rate Poisson process.
In \cite{Sauerwald10_async_sync_push}, it is shown that the asynchronous version of \push has the same broadcast time as standard \push on regular graphs.
In~\cite{GiakkoupisNW16,AngelMP17}, tight bounds are given for the relation between the broadcast times of synchronous and asynchronous \ppull.

On the random walk literature, there has been some previous work on models related to \meet, motivated mainly by the study of the spread of infectious diseases.
The earliest work considering a process equivalent to \meet is \cite{DimitriouNS06}, which studies general graphs.
It shows that the broadcast time of \meet is at most $O(\log n)$ times larger than the meeting time of two random walks in the graph, and that this upper bound is tight.
Later, the authors of~\cite{CooperFR09} studied \meet for the case of random regular graphs and $k\leq n^\epsilon$ random walks.
They showed that the expected broadcast time is $O(n\log k/k)$.
In~\cite{PPPU10}, the $2$-dimensional finite grid was studied and a broadcast time of $\tilde\Theta(n/\sqrt{k})$ was shown for $k$ random walks.
This work was extended to $d$-dimensional grids in~\cite{LamLMSW12}, where a tight lower bound up to a polylogarithmic factor was also shown.

The continuous variant of \meet in the infinite grid was studied in \cite{kesten2005,kesten2008}.
In these works the initial number of agents at each vertex is a Poisson random variable, with constant mean, and initially the information is placed at the origin.
The authors prove a theorem for the asymptotic shape formed by the set of informed agents.
A similar process is the frog model, where only the informed agents move, while the uninformed ones stay put until they are hit by an informed agent.
This process has been studied for infinite grids \cite{popov2003frogs,alves2002} and finite $k$-ary trees~\cite{hermon2018frogs}.

Other superficially related processes include coalescing random walks \cite{BerenbrinkGKM16,KanadeMS19}, and
coalescing branching walks \cite{MitzenmacherRR18, BerenbrinkGK18}.
See also \cite{Cooper11} for a survey on multiple random walks.

\section{Protocol Descriptions}
\label{sec:processes}

We compare four information spreading protocols.
The first two, \push and \ppull, are standard versions of randomized rumor spreading.
The other two, \visit and \meet, use a system of interacting agents performing independent random walks, and are less standardized.
In \push and \ppull, information is communicated between adjacent vertices, whereas in \visit and \meet information is passed between an agent and a vertex it visits, or between two agents when they meet.
All protocols proceed in a sequence of synchronous rounds.
They are applied on a connected undirected graph $G=(V,E)$ with $|V|=n$ vertices, and the information originates from an arbitrary source vertex $s\in V$.

\paragraph{Push.}
In round zero, vertex $s$ becomes informed.
In each round $t \geq 1$, every vertex $u$ that was informed in a previous round samples a random neighbor $v$ to send the information to, and if $v$ is not already informed, it becomes informed in this round.
We denote by $\Tpush(G,s)$ the number of rounds before all vertices are informed.

\paragraph{Push-Pull.}
As in \push, vertex $s$ is informed in round zero.
In each round $t \geq 1$, every vertex $u\in V$ (informed or not) samples a random neighbor $v$ to exchange information with, and if exactly one of $u$ and $v$ was informed before round $t$, then the other vertex becomes informed as well.
The number of rounds before all vertices are informed is denoted $\Tppull(G,s)$.

\paragraph{Visit-Exchange.}

Let $A$ be a set of \emph{agents}.
Every agent $g\in A$ performs an independent simple random walk on $G$, starting from a vertex sampled independently from the stationary distribution (i.e., each vertex $v$ is sampled with probability $\deg(v)/(2|E|)$).
In round zero, vertex $s$ becomes informed, and every agent that is on vertex $s$ becomes informed as well.
In each subsequent round $t\geq 1$, all agents do a single step of their random walk in parallel.
If an agent that was informed in a previous round visits a vertex $v$ that is not yet informed, then $v$ becomes informed in this round.
Also, if an agent $g$ that is not yet informed visits a vertex which got informed either in a previous round or in the current round (by some other informed agent), then $g$ becomes informed as well.
We denote by $\Tvisit(G,s)$ the number of rounds before all \emph{vertices} (and thus all agents) are informed.

\paragraph{Meet-Exchange.}

As in \visit, a set $A$ of agents perform independent random walks starting from the stationary distribution.
In round zero, all agents that are on vertex $s$ become informed.
If there is no agent on $s$ in round zero, then the first agent to visit $s$ after round zero becomes informed (if more than one agents visit $s$ simultaneously, they all get informed).
After that point, vertex $s$ does not inform any other agent that visits $s$.
In each subsequent round $t$, whenever two agent $g,g'$ meet and exactly one of them was informed in a previous round, the other agent becomes informed as well.
We denote by $\Tmeet(G,s)$ the number of rounds before all \emph{agents} are informed.

If $G$ is a bipartite graph, then, depending on the initial positions of the agents, it is possible that some agents are never informed, thus $\Tmeet(G,s)=\infty$.
To avoid this complication we will sometimes assume that the random walks of the agents are lazy, i.e., a walk stays put in a round with probability $1/2$.
This ensures that $\E{\Tmeet(G,s)} < \infty$, for any connected graph $G$.

\bigskip

We will collectively refer to $\Tpush(G,s)$, $\Tppull(G,s)$, $\Tvisit(G,s)$,
and $\Tmeet(G,s)$ as the \emph{broadcast time} of the corresponding protocol.
We will sometimes omit graph $G$ and source vertex $s$ in this notation, when they are clear from the context.

\section{Examples}
\label{sec:cexamples}

In this section, we provide examples demonstrating that \push or \ppull rumor spreading, \visit, and \meet can have very different broadcast times on the same graph.
More precisely, we present graphs where rumor spreading takes polynomial time while \visit and \meet need only logarithmic time (Sections~\ref{sec:star} and~\ref{sec:double-star}), and also graphs where the converse is true (Sections~\ref{sec:tree} and~\ref{sec:twotrees}).
We demonstrate a similar separation between \visit and \meet (Sections~\ref{sec:twotrees} and~\ref{sec:cycle-stars-cliques}), but the gap is polynomial only in one direction, while in the other it is logarithmic.
We do not know whether there exist graphs where \visit is faster than \meet by more than a logarithmic factor.
In all examples below, we assume that the number of agents is
$ |A| = \alpha n = \Theta(n) $.

\subsection{Star Graph}
\label{sec:star}

Let $S_n$ denote an \emph{$n$-leaf star}, that is, a tree with one internal node (the center of the star), and $n$ leaves; see Fig.~\ref{fig:examples}(a) for an illustration.
This is an example of a graph where \push is very slow, whereas all other processes are very fast.

\begin{lemma}
    \label{lem:star}
    For the graph $S_n$ described above and any source vertex $s$,
    (a)~$\Exp{\Tpush} = \Omega(n\log n)$,
    (b)~$\Tppull \leq 2$,
    (c)~$\Tvisit = O(\log n)$, w.h.p.,
    and
    (d)~$\Tmeet = O(\log n)$, w.h.p.
\end{lemma}
\begin{proof}
\textbf{(a):}
This bound is well-known.
It follows from the observation that the center needs to sample each of the leaves (except possibly for one) before all vertices are informed.
The time for that is the time needed to collect all $n$ coupons (except possibly for one) in a coupon collector's problem, which is $\Theta(n\log n)$ 
in expectation.

\medskip

\textbf{(b):}
This bound is also well-known (and trivial).
It takes one round to inform all vertices if $s$ is the source, and two rounds if $s$ is a leaf.

\medskip

\textbf{(c):} For any pair of vertices $v,u$, the probability that an agent located at $v$ visits $u$ within the next two rounds is at least $1/n$.
Since agents do independent random walks, it follows from standard Chernoff bounds (Theorem~\ref{thm:chernoff}) that, for any placement of the agents at round $t$, at least one of the $|A|=\Theta(n)$ agents will visit a given vertex $u$ by round $t + O(\log n)$ w.h.p.
By this observation, it takes $O(\log n)$ rounds w.h.p.\ until the first agent gets informed (by visiting $s$).
If $s$ is not the center, then the center gets informed in the next round.
After that it takes at most two rounds before all agents are informed, because an agent visits the center every other round.
Finally, every leaf $u$ gets informed in an additional $O(\log n)$ rounds w.h.p., by the same observation we used above.

\medskip

\textbf{(d):} Since the graph is bipartite, we assume that the random walks are lazy (i.e., in every round, each random walk stays put with probability $1/2$).
Similarly to~(c),
for any pair $v,u$, the probability that an agent located at $v$ visits $u$ within the next two rounds is at least $1/(4n)$, thus for any placement of the agents at round $t$, at least one agent visits  $u$ by round $t + O(\log n)$ w.h.p.
It follows that it takes $O(\log n)$ rounds w.h.p.\ until the first agent gets informed (by visiting $s$);
let $g^\ast$ denote that agent (or one of them, if there are many).
We complete the proof by arguing that within an additional $O(\log n)$ rounds, w.h.p.\ every agent $g\neq g^\ast$ meets with $g^\ast$ at the center vertex, and thus, all agents become informed within $O(\log n)$ rounds w.h.p.
This follows from the observation that for any given placement of $g^\ast$ and $g$, the probability they are both at the center vertex in the next round is exactly $1/4$.
Thus, a Chernoff bound yields that $g^\ast$ and $g$ will meet w.h.p.\ within $O(\log n)$ rounds.
\end{proof}


\subsection{Double Star}
\label{sec:double-star}

In the star example above only the \push version of randomized rumor spreading is slow, while \ppull is extremely fast.
Next we present a graph where \ppull (and thus, \push) is slow, while \visit and \meet are fast.
Let $S^2_n$ denote a \emph{double-star} graph: two star graphs with $n / 2$ vertices with their centers connected by an edge; see Fig.~\ref{fig:examples}(b).
\begin{lemma}
    \label{lem:double-star}
    For the graph $S^2_n$ described above and any source vertex $s$,
    (a)~$\Exp{\Tppull} = \Omega(n)$,
    (b)~$\Tvisit = O(\log n)$, w.h.p.,
    and
    (c)~$\Tmeet = O(\log n)$, w.h.p.
\end{lemma}

\begin{proof}
    \textbf{(a):}
        Let $a,b$ be the centers of the two stars.
        For \ppull to complete, $a$ must sample $b$ or $b$  must sample $a$,  at least once.
        The probability of that happening in a given round is at most $2 / (n/2)$.
        Thus, the expected number of rounds until \ppull completes is at least $(n/2)/2$.

    \medskip

    \textbf{(b):}
        Let $\EE_u(t)$ denote the event that at least $|A|/8$ agents visit vertex $u\in \{a,b\}$ in round~$t$.
        We consider the following modification to process \visit.
        \begin{quote}
        \textbf{Modification 1:}
            For any round $t\geq0$ and $u\in\{a,b\}$, if event $\EE_u(t)$ does not hold, then  before round $t+1$ we add a number of new and informed agents to the graph, at node $u$, such that  there are $|A|/8$ agents at $u$.
        \end{quote}
        In \visit, at any round $t$, the expected number of agents that visit $u$ is greater than $|A| / 4$.
        It follows, $\Pr{\EE_u(t)} \geq 1 - e^{-\Omega(|A|)} = 1 - e^{-\Omega(n)}$ by a Chernoff bound.
        By applying a union bound for each $u\in\{a,b\}$ and round $t \leq \log^2 n$, we get that, with probability at least $1 - e^{-\Omega(n)}$, the modified process is identical to the original \visit for the first $\log^2 n$ rounds.
        Since our goal is to prove that $\Tvisit = O(\log n)$ w.h.p., it suffices to analyze the modified process.

        In the modified process, since there is at least a linear number of agents at each $u\in\{a,b\}$ before each round, it is straightforward to show that, w.h.p.:
        if $s\notin\{a,b\}$ and $s$ is adjacent, say, to $a$, it takes $O(\log n)$ rounds before $a$ gets informed
        (if $s = a$, $a$ is informed at round zero);
        then in $O(\log n)$ additional rounds $b$ gets informed;
        and finally in $O(\log n)$ extra rounds all leaves are informed.

        \medskip

    \textbf{(c):}
        We assume that the walks are lazy, as the graph is bipartite.
        We apply to \meet the same modification we made to \visit in part (b).
        We also make a second modification.
        Let $\EE_u'(t)$ denote the event that at least one of the agents at vertex $u\in \{a,b\}$ stays put in round~$t$.
        \begin{quote}
            \textbf{Modification 2:}
            For any round $t\geq 0$ and $u\in\{a,b\}$, if event $\EE_u'(t)$ does not hold, then  before round $t+1$ we add a new and informed agent to the graph, at node $u$.
        \end{quote}
        Once again, it is easy to show that with probability at least $1 - e^{-\Omega(n)}$, the modified process is identical to \meet in the first $\log^2 n$ rounds, thus we can analyze the modified process.

        Similarly to part (b), we have that the following hold w.h.p.\ for the modified process.
        If $s\notin\{a,b\}$ and $s$ is adjacent, say, to $a$, it takes $O(\log n)$ rounds before some agent visits $s$, thus gets informed, and then visits $a$.
        From that point on, by our second modification, there is always some informed agent at $a$.
        Then in $O(\log n)$ additional rounds some informed agent visits $b$, and again there is always an informed agent at $b$, thereafter.
        Finally, in $O(\log n)$ extra rounds every agent that is not already informed visits one of $a,b$ and thus gets informed.
\end{proof}



\subsection{Heavy Binary Tree}
\label{sec:tree}
Next
we describe a graph where \visit is slow, while the other processes are fast.
Let $B_n$ denote a \emph{heavy binary tree}, which is constructed by adding an edge between every pair of leaves of a balanced binary tree with $n$ vertices.
Even though $B_n$ is not a tree, we will refer to the leaves of the original binary tree as the leaves of $B_n$.
The set of leaves of $B_n$ induces a clique of $l = \ceil{n/2}$ vertices.
See Fig.~\ref{fig:examples}(c) for an illustration.

\begin{lemma}
    \label{lem:tree}
    For the graph $B_n$ described above and any source vertex $s$,
    (a)~$\Tpush = O(\log n)$, w.h.p., and
    (b)~$\Exp{\Tvisit} = \Omega(n)$.
    If the source $s$ 
    is a leaf,
    then (c)~$\Tmeet = O(\log n)$, w.h.p. 
\end{lemma}

\begin{proof}
\textbf{(a):}
First, we bound the number of rounds until some internal node is informed.
This is zero if $s$ is an internal node, so suppose $s$ is a leaf.
The number of rounds before all leaves are informed is $O(\log n)$ w.h.p.
This follows from the well-known logarithmic bound on the \push broadcast time on a clique,
and the fact that random failures of transmission with probability $1 / l$ (corresponding to the case when a leaf samples its parent)
do not change the broadcast time asymptotically~\cite{elsasser2009runtime}.
Once all leaves are informed, it takes at most $O(\log n)$ additional rounds, w.h.p., until the first internal node is informed, because
there are $l$ leaves and, in each round, each leaf samples its parent with probability $1 / l$.
Once some internal node becomes informed, then all internal nodes become informed after at most $O(\log n)$ rounds w.h.p.
This follows from the observation that the broadcast time of \push on $B_n$ starting from an internal node is dominated by the broadcast time on a balanced binary tree with $n$ vertices.
Since the binary tree has bounded degree and logarithmic diameter, the broadcast time of \push is $O(\log n)$ w.h.p.~\cite{Feige1990}.
Adding all these logarithmic bounds and applying a union bound proves~(a).
\medskip

\textbf{(b):}
Since agents are initially distributed according to the stationary distribution, it follows that a given agent visits the root vertex with probability $2/(2|E|) \leq 8/n^2$ at any given round.
Therefore, the expected number of times agents visit the root during the first $n^2/(16|A|)$ rounds of \visit is at most $1/2$.
It follows that with probability at least $1/2$ no agent visits the root in any of the rounds $t$, $0\leq t < n^2/(16|A|) = \Theta(n)$.
From this it is immediate that the expected number of rounds before the first agent visits the root is at least $\Omega(n)$; this implies~(b).

\medskip

\textbf{(c):}
    Let $\EE(t)$ denote the event that at most $r = c\log n$ agents visit internal nodes at round $t$, where $c > 0$ is a large enough constant.
    We apply the following modification to \meet.
    \begin{quote}
        \textbf{Modification:}
        For any round $t \geq 0$, if $\EE(t)$ does not hold, then before round $t+1$ we move all agents that are at internal nodes  to leaf nodes. (It is not important to which leaves we move the agents.)
    \end{quote}
Since the random walks of the  $|A|=\Theta(n)$ agents start from the stationary distribution, the expected number of agents that visit internal nodes at any given round $t$ is $O(1)$.
Furthermore, since the random walks are independent, a Chernoff bound gives that event $\EE(t)$ holds w.h.p.\ (where the probability is controlled by the choice of $c$).
By a union bound, event $\bigcap_{0\leq t< \log ^2n} \EE(t)$ holds also w.h.p.
It follows that w.h.p.\ the modified process is identical to the original one in the first $\log^2 n$ rounds.
Next we analyze this modified process.

Let $t^\ast\geq 0$ be the first round when some agent visits source $s$, and let $g^\ast$ be an agent that visits $s$ in that round, and thus gets informed.
We have that $t^\ast = O(\log n )$ w.h.p., because by the modification above, there are $\Omega(n)$ agents on leaf nodes before each round, thus the probability at least one agent visits leaf $s$ in any given round is $\Omega(1)$.

For each $g\in A$, we denote by $t_g$ the round when $g$ gets informed.
In particular, $t_{g^\ast} = t^\ast$.
Also, let $I_t = \{g\colon t_g \leq t\}$ be the set of informed agents after round $t$.

Next we show that at least $2r$ agents are informed by some round $t^\ast + O(\log n)$.

\begin{claim}
    \label{clm:logfirst}
    W.h.p.,
    $\min\{k\colon |I_{t^\ast + k}| \geq 2r \} = O(\log n)$.
\end{claim}
\begin{proof}
Recall that $\alpha = |A| / n$ is a constant, and let
$
    r' = 5r / \alpha = \Theta(\log n)
    .
$
For any agent $g$, let $\EE_g$ be the event that $g$ visits only leaf vertices in rounds $t^\ast+1,\ldots, t^\ast + r'$.
Suppose that $g$ is at a leaf before round $t^\ast+1$.
Then
\[
    \Pr{\EE_g} = (1 - 1/l)^{r'} \geq 1 - r' / l.
\]
Also,
\[
    \Pr{t_g \leq t^\ast + r' \mid \EE_g,\EE_{g^\ast}}
    \geq
    1 - \left(1 - \frac{l-2}{(l-1)^2}\right)^{r'}
    \geq
    \frac{r'}{2(l - 1)},
\]
where $(l-2)/(l-1)^2$ is the probability that $g$ and $g^\ast$ visit the same leaf at a given round, assuming that they are at different leaves before the round, and that they both visit leaves at that round.
Using the fact that at least $|A| - r - 1$ agents $g\neq g^\ast$ are on leaves before round $t^\ast + 1$ (due to the modification above),
we obtain for the number of informed agents after round $t^\ast + r'$,
\[
    \Exp{|I_{t^\ast + r'}| \mid \EE_{g^\ast}}
    \geq
    1 + (|A| - r - 1) \cdot \left(1 - \frac{r'}{l}\right) \cdot \frac{r'}{2(l - 1)}  \geq 1 + 4r,
\]
where the extra $1$ accounts for $g^\ast$.
We can thus apply a Chernoff bound to obtain
\[
    \Pr{|I_{t^\ast + r'}| \geq 2r \mid \EE_{g^\ast}}
    \geq
    1 - 1/n,
\]
for $c$ large enough.
From that and
$\Pr{\EE_{g^\ast}} \geq 1 - r' / l = 1 - O(\log n /n)$, it follows
\begin{equation}
    \label{eq:itast}
    \Pr{|I_{t^\ast + r'}| \geq 2r} = 1 - O(\log n / n).
\end{equation}

We can amplify the above probability as follows.
Suppose that  $|I_{t^\ast + r'}| < 2r$.
Consider the first round $t' \geq t^\ast + r'$ such that $g^\ast$ is at a leaf vertex before round $t'+1$.
Then $t' = t^\ast + r' + O(\log n)$, w.h.p.
The reason is that from any internal vertex, an agent reaches a leaf after at most $O(\log n)$ rounds w.h.p., by the properties of a biased random walk on the line \cite[Section~14.2]{feller_introduction_2009}, as the probability of the agent moving closer to the root in a round is $1/3$, while the probability of moving closer to the leaf level is $2/3$.

We can now apply the same argument as in the proof of~\eqref{eq:itast}, using $t'$ in place of $t^\ast$, to obtain
$
    \Pr{|I_{t' + r'}| \geq 2r \mid |I_{t'}| < 2r} = 1 - O(\log n / n).
$
Repeating the argument a constant $i$ number of times, we obtain that
$
    \Pr{|I_{t' + r''}| \geq 2r} = 1 - O(\log n / n)^i,
$
for some $r'' = \Theta(\log n)$.
\end{proof}

Next we argue that once $2r$ agents have been informed, at least half of the agents (or $n/2$ if $|A|>n$) are informed after $O(\log n)$ additional rounds.

\begin{claim}
    \label{clm:nhalf}
    There is a constant $\epsilon >0$, such that
    if $2r \leq |I_t| \leq \min\{n,|A|\}/2$, then
    \[
        \Pr{|I_{t+1}| \geq (1+\epsilon)\cdot |I_t| \mid I_t}\geq 1/2.
    \]
\end{claim}
\begin{proof}
    Suppose that $|I_t| = k \in [2r, \min\{n,|A|\}/2 ]$.
    By the modification we have made, at least $k - r \geq k / 2$ informed agents are on leaf nodes before round $t+1$; let $B$ be the set of these agents.
    Let $L$ be the set of leaves visited by at least one informed agent in round $t+1$.
    By a Chernoff bound,
    \[
        \Pr{|L| \geq k/8} = 1- e^{-\Omega(k)},
    \]
    because for each agent $g$ among \emph{the first $k/2$} agents in $B$, the probability that  in round $t+1$, $g$ visits a leaf that no other agents among the first $k/2$ agents in $B$  visit in the round, is at least $1 - (k/2)/l \geq 1/2$.

    Given $|L|$,
    consider an agent $g$ which is at a leaf before round $t + 1$ and is not yet informed.
    The probability that $g$ visits a leaf in $L$ in round $t + 1$, and thus gets informed, is at least $|L| / l$.
    There are at least $|A| - r - k$ such agents, and therefore,
    the expected number of agents that get informed in round $t + 1$ is
    at least
    $(|A| - r - k) \cdot |L| / l \geq 16\epsilon |L|$ for a sufficiently small constant $\epsilon > 0$.
    Since the agents move independently, by a Chernoff bound we obtain
    \[
        \Pr{|I_{t+1}| \geq k + 16\epsilon |L| / 2 \mid |L| \geq k/8}
        \geq
        \Pr{|I_{t+1}| \geq (1 + \epsilon)k \mid |L| \geq k/8}
        =
        1 - e^{-\Omega(k)}.
    \]
    The claim then follows by combining the two equations we have shown above.
\end{proof}

By applying Claim~\ref{clm:nhalf} repeatedly, for a logarithmic number of rounds, we obtain that if $|I_t| \geq 2r$, then w.h.p,
\[
    \min\{k\colon |I_{t+k}| \geq  \min\{n,|A|\}/2\} = O(\log n).
\]

Next we argue that once $\min\{n,|A|\}/2$ agents have been informed, the remaining agents are informed after $O(\log n)$ additional rounds.

\begin{claim}
    \label{clm:all}
    If $|I_t| \geq \min\{n,|A|\}/2$ and $t_g > t$, then $t_g = t + O(\log n)$ w.h.p.
\end{claim}
\begin{proof}
    We saw in the proof of Claim~\ref{clm:logfirst}, that if $g$ is on an internal node after round $t$, it  will reach a leaf after at most $O(\log n)$ rounds w.h.p.
    Suppose now that $g$ is at a leaf vertex before round $t'+1$, for some $t' \geq t$.
    As we saw earlier, the probability that $g$ visits leaves in all rounds $t'+1,\ldots, t' + r'$, where $r' = \log n$, is at least $1 - r' / l$.
    For a given round in which $g$ visits a leaf, let $q$ be the probability that no informed agent visits the same leaf.
    Since there are at least $\min\{n,|A|\}/2 - r$ informed agents at leaf vertices before each round,
    \[
        q
        \leq
        \frac{1}{l - 1} + \left(1 - \frac{1}{l} \right)^{\min\{n,|A|\}/2 - r}
        \leq
        \beta
        <1,
    \]
    for a constant $\beta$ that depends on $\alpha$.
    This bound follows from the observation that
    $q$ is maximized when all $\min\{n,|A|\}/2 - r$
    informed agents are on the same leaf before the round.
    It follows
    \[
        \Pr{t_g > t' + r'}
        \leq
        r' / l +
        q^{r'}
        =
        O(n^{-\gamma}),
    \]
    for some constant $\gamma > 0$.
    By repeating the argument a constant number of times we obtain the claim for an arbitrary high probability.
\end{proof}

Combining all the above results we complete the proof of (c).
\end{proof}

\subsection{Siamese Heavy Binary Trees}
\label{sec:twotrees}
We
consider now an example where both random walk based processes are slow, while rumor spreading is fast.
Let $D_n$ denote a graph obtained by taking two copies of the graph $B_n$ described above and merging the two roots into a single root vertex; see Fig.~\ref{fig:examples}(d).

\begin{lemma}
    \label{lem:twotrees}
    For the graph $D_n$ described above and any source vertex $s$,
    (a)~$\Tpush = O(\log n)$, w.h.p.,
    (b)~$\Exp{\Tvisit} = \Omega(n)$,
    and
    (c)~$\Exp{\Tmeet} = \Omega(n)$.
\end{lemma}
\begin{proof}
Parts (a) and (b) follow from the same arguments used to prove the corresponding bounds in Lemma~\ref{lem:tree}.
For~(c), we observe that w.h.p.\ at least one agent will start from each of the two trees.
Then, for the information to pass from agents on the one tree to agents on the other, some agent must reach the root, which requires $\Omega(n)$ rounds in expectation, as we showed in the proof of  Lemma~\ref{lem:tree}(b).
\end{proof}

\subsection{Cycle of Stars of Cliques}
\label{sec:cycle-stars-cliques}

Finally,
we present a graph on which \visit is faster than \meet, by a logarithmic factor.
We note that this graph is (almost) regular, unlike the highly non-regular graphs we considered in the previous sections.
We leave open the question whether there are graphs on which \visit is asymptotically faster than \meet by a polynomial factor.

\begin{lemma}
    \label{lem:cycle-stars-cliques}
    There is a graph $G=(V,E)$ with $|V|=\Theta(n)$ such that for any source vertex $s\in V$,
    (a)~$\Exp{\Tvisit} = O(n^{2/3})$,
    and
    (b)~$\Exp{\Tmeet} = \Omega(n^{2/3} \log n)$.
\end{lemma}

\begin{proof}[Proof Sketch]
An example of a graph $G$ with the above properties is a \emph{cycle-of-stars-of-cliques}, obtained as follows:
Consider a cycle graph of length $n^{1/3}$, consisting of vertices $c_i$, $i\in\{1,\ldots,n^{1/3}\}$.
For each $i$ consider a new set of $n^{1/3}$ vertices $l_{i,j}$, $j\in\{1,\ldots,n^{1/3}\}$, and connect $c_i$ to each $l_{i,j}$.
Finally, for each $j$ consider a new set of $n^{1/3}$ vertices $q_{i,j,k}$, $k\in\{1,\ldots,n^{1/3}\}$, add an edge between each pair $q_{i,j,k},q_{i,j,k'}$, and also between $l_{i,j}$ and all $q_{i,j,k}$.
See Fig.~\ref{fig:examples}(e) for an illustration of this graph.
We denote by $Q_{i,j}$ the $(n^{1/3}+1)$-clique induced by the vertex set $\{l_{i,j}\}\cup \{q_{i,j,1}\ldots q_{i,j,n^{1/3}}\}$.

The core-idea is that since vertices $c_i$ are not informed in $\meet$, the information advances from $c_i$ to its neighboring ring vertices $c_{i-1}$ and $c_{i+1}$ slower than in $\visit$.
Below we give a sketch of the analysis.
To make it rigorous, one needs to use techniques similar to those in the other proofs of the paper, namely, bounding above and below the number of agents at subgraphs of $G$.
The number of rounds we refer to below are all in expectation.

\textbf{(a):}
Suppose that the source vertex $s$ is in clique $Q_{i,j}$.
Then it takes $O(\log n)$ rounds until all vertices of the clique are informed.
After that, vertex $c_i$ gets informed in $O(n^{1/3})$ additional rounds, which is the average time it takes for the first agent to cross the edge from $l_{i,j}$ to $c_i$, since a constant number of agents visit each vertex on average.
From $c_{i}$, the information passes to $c_{i-1}$ and $c_{i+1}$ in $O(n^{1/3})$ rounds after $c_i$ is informed.
Thus, it takes $O(n^{2/3})$ rounds before all ring nodes $c_{i'}$ are informed.
Once $c_{i'}$ is informed, it takes $O(n^{1/3}\log n)$ rounds (by coupon collector's) until all cliques $Q_{i',j'}$ are informed.
It follows that the total broadcast time is $O(n^{2/3})$.\footnote{Alternatively, one can prove the statement assuming \push instead of \visit, and then apply Theorem~\ref{thm:intro}, since graph $G$ is (almost) regular.}


\textbf{(b):}
Suppose again that the source $s$ is in clique $Q_{i,j}$.
We first lower bound the number of rounds until at least $\Omega(n^{1/3})$ informed agents visit $c_i$, which is the average number of agents until one of them moves to either $c_{i-1}$ or $c_{i+1}$.
It takes $\Omega(n^{1/3})$ rounds until the first  informed agent visits $c_i$.
This agent will move to another clique $Q_{i,j'}$ with probability $1 - O(n^{-1/3})$.
After that, the next informed agent visiting $c_i$ can come from $Q_{i,j}$ or $Q_{i,j'}$, and, therefore, the expected number of rounds until such a visit is halved.
In general once $\ell$ of the $n^{1/3}$ cliques $Q_{i,*}$ have received an informed agent, $c_i$ is visited by informed agents at the rate of once every $n^{1/3}/\ell$ rounds.
It follows that it takes $\Theta(n^{1/3}\cdot \log n)$ rounds before $c_i$ has been visited by $\Omega(n^{1/3})$ informed agents, and therefore, at least that many rounds are necessary until an informed agent moves to either $c_{i-1}$ or $c_{i+1}$.
Therefore, it takes $\Omega(n^{2/3}\cdot\log n)$ rounds before all nodes on the ring are informed.
\end{proof}

\section{Bounding \texorpdfstring{$\Tpush$}{Tpush} by \texorpdfstring{$\Tvisit$}{Tvisitx} on Regular Graphs}
\label{sec:push<msg}

In this section, we prove the following theorem, which upper bounds the broadcast time of \push in a regular graph by the broadcast time of \visit.

\begin{theorem}
    \label{thm:Tpush<Tmsg}
    For any constants $\varepsilon,\alpha,\lambda > 0$, there is a constant $c>0$, such that
    for any  $d$-regular graph $G = (V, E)$ with $|V| = n$ and $d \geq \varepsilon\log n$, and for any source vertex $s\in V$, the broadcast times of \push and \visit, with $|A| \leq \alpha  n$ agents, satisfy
    \[
        \Pr{\Tpush \leq ck} \geq \Pr{\Tvisit \leq k} - n^{-\lambda},
    \]
    for any $k\geq 0$.
\end{theorem}

From Theorem~\ref{thm:Tpush<Tmsg}, it is immediate that if $\Tvisit \leq T$ w.h.p., then $\Tpush = O(T)$ w.h.p.
Moreover, using Theorem~\ref{thm:Tpush<Tmsg} and the known $O(n\log n)$ upper bound on $\Tpush$ which holds w.h.p.~\cite{Feige1990}, one can easily obtain that $\Exp{\Tpush} = O(\Exp{\Tvisit})$.


\paragraph{Proof Overview of Theorem~\ref{thm:Tpush<Tmsg}.}

The proof uses the following coupling of processes \push and \visit:
For each vertex $u$, let $\langle\pi_u(1),\pi_u(2),\ldots\rangle$ be the sequence of neighbors that $u$ samples in \push after getting informed.
Similarly, for \visit, consider all moves of informed agents from $u$ to its neighbor vertices in chronological order, and let $\langle p_u(1),p_u(2),\ldots\rangle$ be the destination vertices in those moves (we order moves in the same round by, say, agent ID).
We couple the two processes by setting $\pi_u(i)=p_u(i)$, for all $u,i$.

The intuition for this coupling is that in \visit, at most a constant number of agents \emph{in expectation} visits each vertex $u$ in a round (since the graph is regular and $|A|=O(n)$), and thus the same number of agents leaves $u$ per round in expectation.
The coupling ensures that for each informed agent that moves from $u$ to a neighbor $v$, vertex $u$ samples the same neighbor $v$ in \push.
Thus, if we had a constant upper bound $c$ on the \emph{actual} number (rather than the expected number) of visits to each vertex on each round, then the coupling would immediately yield $\Tpush \leq c\cdot \Tvisit$ for the coupled processes.
In reality, however, a super-constant number of agents may visit a vertex in a round, and, moreover, the number of visits depends on the past history of the process.

An basic idea we use to tackle dependencies on the past history is to consider a tweaked version of \visit, called \tvisit.
The only difference between this process and \visit, is that it arbitrarily removes some agents after each round to ensure that the neighborhood of any vertex contains at most $O(d)$ agents.
For $d = \Omega(\log n)$ and $|A| = O(n)$, we have that in the first $\mathrm{poly}(n)$ rounds the two processes are identical w.h.p. Therefore, we can consider \tvisit in our proofs.
The benefit we get is that since the neighborhood of any vertex $u$ contains $O(d)$ agents in round $t$, at round $t+1$ the number of agents that visit $u$ will be bounded by the binomial distribution $\Bin(\Theta(d), 1/d)$, \emph{independently of the past}.

To prove the theorem is suffices to show that under our coupling,
with probability at least $1 - n^{-\lambda}$, if $\Tvisit \leq k$ then $\Tpush \leq ck$.
Further, we will assume that $k$ is at least $\Omega(\log n)$;
for $k = O(\log n)$ the theorem is obtained by showing that $\Tvisit = \Omega(\log n)$ w.h.p.

To show that w.h.p.\ $\Tvisit \leq k$ implies $\Tpush \leq ck$, we consider all possible paths of length $k$ through which information travels in \visit, and for each path we count the total number of (non-distinct) agents encountered along this path, called the congestion of the path.
Formally, we use the notion of a \emph{canonical walk} $\theta$, which is represented by a sequence of vertices $\theta = (\theta_0,\theta_1,\ldots,\theta_k)$ starting from $\theta_0 = s$:
In each round $1\leq t\leq k$, the walk either stays put and $\theta_t = \theta_{t-1}$, or it follows one of the agents $g$ that leave $\theta_{t-1}$ in round $t$, and, in that case, $\theta_{t}$ is the new vertex that $g$ moves to.
For any round $t$, we count the agents that are in $\theta_t$.
The sum of these counts, for $0 \leq t < k$ is the \emph{congestion} $Q(\theta)$ of the walk $\theta$.

The congestion of a canonical walk is used to bound the time needed for information to travel along the \emph{same} path in the coupled \push process.
Intuitively, larger congestion implies longer travel time for \push, for the following reason.
Suppose there are $m$ agents in $u$ at some round after it is informed by \visit.
The coupled \push process, using the same random decisions for the choice of neighbors as \visit, will take $m$ rounds to ``go through'' these $m$ agents.

To relate the congestion of canonical walks with the time it takes for information to spread in \push, we introduce \emph{C-counters}:
%
%
For each vertex $u$, we maintain a counter $C_u$.
The counter is initialized in the round $t_u$ in which $u$ becomes informed in \visit.
Its initial value is the value of the C-counter of the neighbor from which the first informed agent arrived to $u$.
In each subsequent round $t>t_u$, $C_u$ increases by the number of agents that visited $u$ in round $t-1$.
C-counters have the following two properties:
If $\tau_u$ is the round when $u$ gets informed in \push then $\tau_u \leq C_u(t_u)$; and for any $t\geq t_u$, there is a canonical walk $\theta$ of length $t$ such that $C_u(t) = Q(\theta)$.
Therefore, to show that w.h.p.\ $\Tvisit \leq k$ implies $\Tpush \leq ck$, it suffices to show that the maximum congestion of \emph{all} canonical walks of length $k$ is at most $ck$ w.h.p.

We can bound the congestion of a \emph{single} canonical walk of length $k$ using the property of \tvisit that the number of agents at a node is bounded by a binomial distribution with constant mean.
This results in the desired bound of $ck$ for a single walk with probability at least $1 - a^{-k}$, for some constant $a>1$.
We would like to take a union bound over all canonical walks, which would give the desired result.
For this to work, however, we should also bound the total number of canonical walks of length $k$ by at most $a^{k}/n^\lambda$.

We bound the number of canonical walks of length $k$ by introducing a set of \emph{descriptors} for these walks.
A descriptor is represented by a matrix, which, together with a given execution of \visit, uniquely defines a canonical walk.
Additionally, the set of descriptors suffices to encode all canonical walks, and therefore, it is at least as large as the set of all walks.
Thus, we can use a bound on the number of descriptors that can be computed by a simple combinatorial argument involving the number of elements used in the matrix, and the values they can take.
A naive construction of descriptors, however, is too wasteful giving us a much larger bound than the $a^{k}/n^\lambda$ we need.
A key idea here is that the majority of the descriptors represent walks only in executions that happen with low probability.
So, we construct a set of \emph{concise descriptors} that can describe all canonical walks in a random execution \emph{w.h.p.}
We show that the size of the set of concise descriptors can be bounded by $a^{k}/n^\lambda$, as desired.
Next we give the details of the proof.

\subsection{Notation and Coupling Description}
\label{sec:push<msg-coupling}

For each vertex $u\in V$, we denote by $\tau_u$ the round when $u$ gets informed in \push.
For $i \geq 1$, let $\pi_u(i)$ be the $i$th vertex that $u$ samples, i.e., the vertex it samples in round $\tau_u + i$.
Note that $\tau_{\pi_u(i)} \leq \tau_u + i$.
In \visit, we denote by $t_u$ the round when vertex $u$ gets informed.
For any agent $g\in A$ and $t\geq 0$, we denote by $x_g(t)$, the vertex that $g$ visits in round $t$.
Thus, $\{x_g(t)\}_{t\geq 0}$ is a random walk on $G$.
Let $Z_u(t)$ be the set of all agents that visit $u$ in round $t$, i.e.,
\[
    Z_u(t) = \{g\in A\colon x_g(t)=u\}.
\]
Thus, $Z_u(t)$ is also the set of agents that depart from $u$ in round $t+1$.
Consider all visits to $u$ in rounds $t\geq t_u$, in chronological order, ordering visits in the same round with respect to a predefined total order over agents.
For each $i \geq 1$, consider the agent $g$ that does the $i$th such visit, and let $p_u(i)$ be the vertex that $g$ visits next. 
Formally, let
$
    X_u =\{(t,g)\colon t\geq t_u, x_g(t) = u\},
$
and order its elements such that $(t,g)<(t',g')$ if $t<t'$, or $t=t'$ and $g<g'$.
If $(t,g)$ is the $i$th smallest element in $X_u$, then $p_u(i) = x_g(t+1)$.

\paragraph{Coupling.}
We couple processes \push and \visit by setting $\pi_u(i) = p_u(i)$.
Formally, let $\{w_u(i)\}_{u\in V,i\geq 1}$, be a collection of independent random variables, where $w_u(i)$ takes a uniformly random value from the set $\Gamma(u)$ of $u$'s neighbors.
Then, for every $u\in V$ and $i\geq 1$, we set
$
    \pi_u(i) = p_u(i) = w_u(i)
    .
$

\subsection{Upper Bound on Agents and Tweaked Visit-Exchange}
\label{sec:tweaked}

We will use the next simple bound on the number of agents that visit a given set $S$ of vertices in some round $t$ of \visit.
The proof is by a simple Chernoff bound, and relies on the assumption that agents execute independent walks starting from stationarity.

\begin{lemma}
    \label{lem:msgdensityub}
    For any $S\subseteq V$, $t\geq0$, and $\beta\geq 2e\cdot |A|/n$,
    \[
        \Pr{\sum_{v\in S}|Z_v(t)| \leq \beta\cdot |S|}
        \geq
        1-2^{-\beta\, |S|}.
    \]
\end{lemma}
\begin{proof}
Since each random walk starts from stationarity, and $G$ is a regular graph, it follows that for any agent $g\in A$, $\Pr{x_g(t)\in S} = |S|/n$.
Thus, the expected number of agents that visit $S$ in round $t$ is
$
|A|\cdot |S|/n \leq \beta\cdot |S|/(2e).
$
Then, by the independence of the random walks, we can use a standard Chernoff bound to show that the number of agents that visit $S$ at $t$ is at most $\beta\cdot |S|$ with probability at least $1-2^{-\beta\cdot |S|}$.
\end{proof}

We remark that Lemma~\ref{lem:msgdensityub} holds also in the case where $|A| = n$ and exactly one walk starts from each vertex.
This implies that Theorem~\ref{thm:Tpush<Tmsg} holds in the above case as well, because the rest of the proof does not require any assumptions about the initial distribution of agents.



In parts of the analysis, we will use a ``tweaked'' variant of \visit, called \tvisit, defined as follows.
Let
\begin{equation}
    \label{eq:a}
    \gamma \geq 2e\cdot |A|/n
\end{equation}
be a (sufficiently large) constant to be specified later.
If in some round $t\geq 0$, there is a vertex $u\in V$ for which the following condition does \emph{not} hold:
\begin{equation}
    \label{eq:alphaCondition}
    \sum_{v \in \Gamma(u)} |Z_v(t)| \leq \gamma\cdot d,
\end{equation}
then before round $t+1$, we remove a minimal set of agents from the graph in such a way that the above condition holds for all vertices $u$, when counting just the remaining agents.

It follows from Lemma~\ref{lem:msgdensityub} that
if constant $\gamma$ is large enough, and $d = \Omega(\log n)$,
then w.h.p.\ the modified process is identical to the original in the first polynomial number of rounds.

\begin{lemma}
    \label{lem:tweacked}
    The probability that Eq.\eqref{eq:alphaCondition} holds simultaneously  for all $u\in V$ and $0 \leq t < k$ is at least $1-kn\cdot 2^{-\gamma d}$.
\end{lemma}
\begin{proof}
The claim follows by applying Lemma~\ref{lem:msgdensityub}, for each $0\leq t < k$ and each pair $u,S$, where $u\in V$ and $S = \Gamma(u)$, and then combining the results using a union bound.
\end{proof}

We use the same definitions and notations for both \visit and \tvisit.

\subsection{C-Counters}

Recall that $t_u$ is the round when vertex $u$ gets informed in \visit.
If $u\neq s$, this is the first round when some informed agent visits $u$.
We are interested in the neighbor $v$ of $u$ from which that agent arrived.
Note that $t_v < t_u$.
Note also that there may be more than one such neighbors $v$, if more than one informed agent visit $u$ at round $t_u$.
For each $u\in V$, let
\[
    S_u = \{v\in \Gamma(u)\colon t_v<t_u,\, Z_v(t_u-1)\cap Z_u(t_u) \neq\emptyset\},
\]
i.e., $S_u$ contains all neighbors $v$ of $u$ for which some informed agent moved from $v$ to $u$ in round $t_u$.
Next, for each $t\geq 0$, we define the counter variable
\begin{equation}
    \label{eq:Cu}
    C_u(t)
    =
    \begin{cases}
    0, & \text{if $t<t_u$ or $t=t_u=0$}
    \\
    \min_{v\in S_u} C_v(t), & \text{if $t=t_u>0$}
    \\
    C_u(t-1) + |Z_u(t-1)|, & \text{if $t>t_u$}.
    \end{cases}
\end{equation}
That is, $C_u$ is initialized in round $t_u$ to the minimum counter value of the neighbors in $S_u$ (or to zero if $u=s$), and $C_u(t)- C_u(t_u)$ is the number of visits to $u$ from round $t_u$ until round $t-1$, or equivalently, the number of departures of agents from $u$ in rounds $t_u+1$ up to $t$.


\begin{lemma}
    \label{lem:tauleqC}
    For any $u\in V$, $\tau_u \leq C_u(t_u)$.
\end{lemma}
\begin{proof}
Consider the following path through which information reaches $u$ in \visit.
The path is $(v_0,v_1,\ldots, v_k)$, where $v_0=s$, $v_k=u$, and for each $0<j\leq k$, we have $v_{j-1}\in S_{v_j}$ and $C_{v_{j-1}}(t_{v_j}) = \min_{v\in S_{v_j}} C_{v}(t_{v_j}) = C_{v_j}(t_{v_j})$.
We prove by induction on $0\leq j\leq k$ that
\begin{equation}
    \label{eq:tauind}
    \tau_{v_j} \leq C_{v_j}(t_{v_j}).
\end{equation}
This holds for $j=0$, because $v_0=s$, $t_s=0$, and $\tau_s=0=C_s(0)$.
Let $0 < j\leq k$, and suppose that $\tau_{v_{j-1}} \leq C_{v_{j-1}}(t_{v_{j-1}})$; we will show that $\tau_{v_j} \leq C_{v_j}(t_{v_j})$.
We have
\begin{align*}
    C_{v_j}(t_{v_j})
    &=
    C_{v_{j-1}}(t_{v_j}),
    \quad
    \text{by the path property}
    \\&
    =
    C_{v_{j-1}}(t_{v_{j-1}}) + \sum_{t_{v_{j-1}}\leq t < t_{v_j}}|Z_{v_{j-1}}(t)| ,
    \quad\text{by recursive application of~\eqref{eq:Cu}}
    \\&
    \geq
    \tau_{v_{j-1}} + \sum_{t_{v_{j-1}}\leq t < t_{v_j}}|Z_{v_{j-1}}(t)|,
    \quad
    \text{by the induction hypothesis.}
\end{align*}
Let $\ell = \min\{i\colon p_{v_{j-1}}(i) = v_{j}\}$, let $g$ be the agent that does the $\ell$th visit to $v_{j-1}$ since round $t_{v_{j-1}}$, and let $r$ be the round when that visit takes place, thus $x_g(r) = v_{j-1}$ and $x_g(r+1) = v_{j}$.
By the minimality of $\ell$, $r+1$ is the first round when some informed agent moves to $v_j$ from $v_{j-1}$.
Since $v_{j-1} \in S_{v_j}$, it follows that $r+1 = t_{v_j}$. Then
\[
    \ell
    \leq
    \sum_{t_{v_{j-1}}\leq t \leq r}|Z_{v_{j-1}}(t)|
    =
    \sum_{t_{v_{j-1}}\leq t < t_{v_j}}|Z_{v_{j-1}}(t)|.
\]
Also, from the coupling, $\pi_{v_{j-1}}(\ell) = p_{v_{j-1}}(\ell) =  v_{j}$, which implies
\[
    \tau_{v_{j}} \leq \tau_{v_{j-1}} + \ell.
\]
Combining all the above we obtain
$
    C_{v_j}(t_{v_j})
    \geq
    \tau_{v_{j-1}} +\ell \geq \tau_{v_j}
    ,
$
completing the inductive proof of~\eqref{eq:tauind}.
Applying~\eqref{eq:tauind} for $j=k$, we obtain $\tau_u \leq C_u(t_u)$.
\end{proof}

\subsection{Canonical Walks and Congestion}

Let $\theta = (\theta_0,\theta_1,\ldots,\theta_k)$,
where $\theta_0 = s$ and $\theta_i\in \Gamma(\theta_{i-1})\cup\{\theta_{i-1}\}$ for $1\leq i \leq k$,
be a walk on $G$ constructed from \visit as follows.
We start from vertex $\theta_0 = s$ in round zero, and in each round $1\leq t \leq k$, we either stay put, in which case $\theta_t = \theta_{t-1}$, or we choose one of the agents $g\in Z_{\theta_{i-1}}(t-1)$, which visited $\theta_{i-1}$ in the previous round, and move to the same vertex as $g$ in round $t$, i.e., $\theta_t = x_t(g)$.
We call $\theta$ a \emph{canonical walk} of length $k$. 
A \emph{labeled canonical walk} is a canonical walk that specifies also the agent $g_t$ that the walk follows in each step $t$, if $\theta_t\neq \theta_{t-1}$.
Formally, a labeled canonical walk corresponding to $\theta$ is $\eta = (\theta_0,g_1,\theta_1,g_2,\ldots,g_k,\theta_k)$, where $g_t \in Z_{\theta_{t-1}}(t-1)\cap Z_{\theta_t}(t)$ if $\theta_{t}\neq \theta_{t-1}$, and $g_t = \bot$ if $\theta_{t}=\theta_{t-1}$.
Note that different labeled canonical walks may correspond to the same (unlabeled) canonical walk.
We define the \emph{congestion} $Q(\theta)$
of a canonical walk $\theta$ as the total number of agents encountered along the walk,\footnote{%
    The same agents is counted more than once if encountered in multiple rounds.
}
not counting the last step, i.e.,
\[
    Q(\theta) = \sum_{0\leq t< k} |Z_{\theta_t}(t)|.
\]
The congestion of a \emph{labeled} canonical walk is the same as the congestion of the corresponding unlabeled walk.

\begin{lemma}
    \label{lem:DeqC}
    For any $u\in V$ and $t\geq t_u$, there is a canonical walk $\theta$ of length $t$ with $Q(\theta) = C_u(t)$.
\end{lemma}
\begin{proof}
We consider the same path $(v_0,v_1,\ldots, v_k)$ as in the proof of Lemma~\ref{lem:tauleqC}, where $v_0=s$, $v_k=u$, and for each $0<j\leq k$, $v_{j-1}\in S_{v_j}$ and $C_{v_j}(t_{v_j}) = C_{v_{j-1}}(t_{v_j})$.
Consider the canonical walk $\theta$ obtained from this path by adding between each pair of consecutive vertices $v_{j-1}$ and $v_j$, $t_{v_{j}}-t_{v_{j-1}}-1$ copies of $v_{j-1}$, and also appending after $v_k$ a number of $t-t_{v_{k}}$ copies of $v_k$.
It is then easy to show by induction that $Q(\theta) = C_u(t)$.
\end{proof}

\subsection{Concise Descriptors of Canonical Walks}

In this section, we bound the number of distinct labeled canonical walks of a given length $k$.
For that, we present a concise description for such walks, and bound the total number of the walks  by the total number of different possible descriptions.

We start with a rather wasteful way to describe labeled canonical walks, which we then refine in two steps.
Let $\AA_k$ denote the set of all $\alpha n\times k$ matrices $A_k = [a_{i,j}]$, where $a_{i,j}\in \{0,\ldots,i\}$.
Let us fix the first $k$ rounds of \visit, and consider a labeled canonical walk
$\eta = (\theta_0=s,g_1,\theta_1,\ldots,g_k,\theta_k)$.
For each $1\leq t \leq k$,
let
\[
    \delta_t = |Z_{\theta_{t-1}}(t-1)|
\]
be the number of agents that visit $\theta_{t-1}$ in round $t-1$, and thus also the number of agents that depart from $\theta_{t-1}$ in round $t$.
Let $\rho_t = 0$ if $g_t = \bot$, otherwise, $\rho_t$ is equal to the \emph{rank} of $g_t$ in set $Z_{\theta_{t-1}}(t-1)$, i.e., $\rho_t = |\{g\in Z_{\theta_{t-1}}(t-1)\colon g \leq g_t\}|$.
We describe walk $\eta$ by a matrix $A_k\in \AA_k$ with the following entries:
For each $1\leq t \leq k$, if $\delta_t>0$, then $a_{\delta_t,j} = \rho_t$, for  $j = |\{t'\leq t\colon \delta_{t'} = \delta_t\}|$,
i.e., value $\rho_t$ is stored in the first unused entry of row $A_k[\delta_t,\cdot]$.
At most $k$ of the entries of $A_k$ are specified that way; the remaining entries can have arbitrary values.
We call $A_k$ a \emph{non-concise descriptor} of $\eta$.

For any given realization of \visit, each $A_k\in\AA_k$ describes exactly one labeled canonical walk of length $k$, and any labeled canonical walk of length $k$ has at least one non-concise descriptor $A_k\in \AA_k$ (in fact, several ones).
%
%
The total number of different non-concise descriptors is $|\AA_k|=\prod_{1\leq i\leq \alpha n} (i+1)^k$, which is too large for our purposes.

A simple improvement is to use only entries in rows $A_k[i,\cdot]$ for which $i$ is a power of 2 (we assume w.l.o.g.\ that $\alpha n$ is also a power of 2).
Roughly speaking, if $\delta_t$ is between $2^{\ell-1}$ and $2^\ell$ then $\rho_t$ is stored in raw $A_k[2^\ell,\cdot]$.
Formally, let $b$ be a (large enough) constant, to be specified later, which is a power of 2.
The matrix $A_k\in \AA_k$ we use to describes~$\eta$ has the following entries.
For each $1\leq t\leq k$:
\begin{enumerate}
  \item If $2^{\ell-1} < \delta_t \leq 2^\ell$, where $\ell\in\{1+\log b,\ldots,\log (\alpha n)\}$, and $|\{t'\leq t\colon 2^{\ell-1} < \delta_{t'} \leq 2^\ell\}|=j$, then
      \begin{enumerate}
        \item if $\rho_t \neq 0$, we have $a_{2^{\ell},j} = \rho_t$,
        \item if $\rho_t = 0$, $a_{2^{\ell},j}$ can take any value in  $\{0\} \cup\{\delta_t+1,\ldots,2^\ell\}$.
      \end{enumerate}
  \item If $0 \leq \delta_t \leq b$  and $|\{t'\leq t\colon 0 < \delta_{t'} \leq b\}|=j$, then
      \begin{enumerate}
        \item if $\rho_t \neq 0$, we have $a_{b,j} = \rho_t$,
        \item if $\rho_t = 0$, $a_{b,j}$ can take any value in $\{0\} \cup\{\delta_t+1,\ldots,b\}$.
      \end{enumerate}
\end{enumerate}
The purpose of subcases (b) is to maintain the property that \emph{every} $A_k$ describes a labeled canonical walk, which would not be the case if we just set $a_{2^{\ell},j} = 0$ or $a_{b,j} = 0$, since values greater than $\delta_t$ would not correspond to a walk.
We call the matrix $A_k$ above a \emph{semi-concise descriptor} of~$\eta$.


A second modification we make is based on the observation that, even in the logarithmic number of $A_k$' rows used in the above scheme, most entries are very unlikely to be actually used.
For each row $i = 2^\ell$, we specify a threshold index $k_i\leq k$, such that the first $k_i$ entries in each row $A_k[i,\cdot]$ suffice  w.h.p.\ to describe all labeled canonical walks of length $k$, in a random realization of \visit.
Let  $\BB_k$ be a subset of $\AA_k$ defined as follows.
Let
\[
    k_i = b\cdot k/i
    ,
\]
and recall that $b$ is a constant power of 2.
The set $\BB_k$ consists of all $A_k=[a_{i,j}]\in\AA_k$ such that
\begin{align*}
    &a_{i,j} \in \{0,\ldots,i\}, &&\text{if $i \in \{2^\ell\colon \log b \leq \ell\leq \log (\alpha n)\}$ and $j \leq k_i$}
    \\
    &a_{i,j} = 0, &&\text{otherwise}.
\end{align*}
A \emph{concise descriptor} of a labeled canonical walk $\eta$ of length $k$ is any semi-concise descriptor $A_k$ of $\eta$ that belongs to set $\BB_k$.

%
%
%
Next we compute an upper bound on the number of all possible concise descriptors of length~$k$.

\begin{lemma}
    \label{lem:cardBk}
    $|\BB_k| \leq  (4b)^{2k}$.
\end{lemma}

\begin{proof} \allowdisplaybreaks
From the definition of $\BB_k$, we have
\begin{align*}
    |\BB_k|
    &\leq
    \prod_{\log b \leq \ell\leq \log(\alpha n)}
    (2^\ell + 1)^{bk/2^\ell}
    \\&
    =
    \prod_{\log b \leq \ell\leq \log(\alpha n)}
    2^{\ell bk/2^\ell}
    \cdot
    \prod_{\log b \leq \ell\leq \log(\alpha n)}
    (1+2^{-\ell})^{bk/2^\ell}
    \\&
    \leq
    \frac
    {\prod_{\ell\geq 1} 2^{\ell bk/2^\ell}}
    {\prod_{\ell \leq \log b - 1} 2^{\ell bk/2^\ell}}
    \cdot
    \prod_{\ell \geq \log b}
    e^{bk/4^\ell}
    \\&
    =
    \frac{2^{2 bk}}
    {2^{(2(b-\log b -1)k}}
    \cdot
    e^{(4/3)k/b}
    \\&
    \leq
    2^{2(\log b + 2)k},
\end{align*}
where in the second-last line we used $\sum_{\ell\geq 1} \ell/2^\ell = 2$, $\sum_{\ell\leq y}\ell/2^\ell = 2^{-y}(2^{y+1}-y-2)$, and $\sum_{\ell\geq 0} 1/4^\ell = 4/3$; and in the last line we used that $e^{(4/3)}<4$.
\end{proof}

For any realization of \visit, each $A_k\in\BB_k$ is a concise descriptor of some labeled canonical walk of length $k$.
However it is not always the case that a labeled canonical walk has a concise descriptor.
The next lemma shows that w.h.p.\ all labeled canonical walks of length $k$ have concise descriptors for an appropriate choice of constant parameter $b$.
Note that the lemma assumes the \tvisit process.
The proof is given in Section~\ref{sec:consisewhp}.

\begin{lemma}
    \label{lem:consisewhp}
    If $b \geq \max\{2\gamma e^2,64\}$ then,
    with probability at least
    $1-2^{-bk/4} \log(\alpha n)$,
    all labeled canonical walks of length $k$
    in a random realization of \tvisit
    have concise descriptors.
\end{lemma}

\subsection{Proof of Lemma~\ref{lem:consisewhp}}
\label{sec:consisewhp}

First, we bound the number of steps $t$ in which more than $i$ agents are encountered in a canonical walk of length $k$.

\begin{lemma}
    \label{lem:distributionzt}
    Fix any $A_k\in\AA_k$, and let
    $\eta = (\theta_0,g_1,\theta_1,\ldots,g_k,\theta_k)$
    be the labeled canonical walk with semi-concise (or non-concise) descriptor $A_k$  in \tvisit.
    For any $i\geq e^2\gamma$ and $\beta \geq e^2 \gamma$,
    \[
        \Pr{|\{t\in\{1,\ldots,k\} \colon  \delta_t > i\}| \geq \beta k/i }\leq 2^{-\beta k}.
    \]
\end{lemma}
\begin{proof}
Recall that $\delta_t = |Z_{\theta_{t-1}}(t-1)|$ is the number of agents that visit vertex $\theta_{t-1}$ in round $t-1$, and thus also the number of agents that depart from $\theta_{t-1}$ in round $t$.
We argue that for any $t\geq 1$, conditioned on $\delta_1,\ldots,\delta_{t}$, variable $\delta_{t+1}$ is stochastically dominated by the binomial random variable $\Bin(\gamma d,1/d)+1$:
From~\eqref{eq:alphaCondition}, applied for vertex $\theta_{t}$ and round $t-1$, we get
\[
    \sum_{v \in \Gamma(\theta_{t})} |Z_v(t-1)| \leq \gamma\cdot d,
\]
thus, there are at most $\gamma d$ agents in the neighborhood of $\theta_{t}$ before round $t$.
If $\theta_t = \theta_{t-1}$,
then each one of those at most $\gamma d$ agents 
will visit $\theta_{t}$ in round $t$ independently with probability $1/d$.
If $\theta_t \neq \theta_{t-1}$ (thus $g_{t} \in Z_{\theta_{t-1}}(t-1) \cap Z_{\theta_{t}}(t)$),
then each of the at most $\gamma d$ agents 
will visit $\theta_{t}$ in round $t$ independently with probability $1/d$, \emph{except} for agent $g_{t}$ who visits $\theta_{t}$ with probability 1.
In both cases, the number $\delta_{t+1}$ of agents that visit $\theta_t$ is dominated by $\Bin(\gamma d,1/d)+1$.
It follows that for any $t\geq 1$ and $i\geq 1$,
\begin{align*}
    \Pr{\delta_{t+1} > i \mid \delta_1,\ldots,\delta_{t}}
    &\leq
    \Pr{\Bin(\gamma d,1/d) + 1> i}
    =
    \Pr{\Bin(\gamma d,1/d) \geq i}
    \\&
    \leq
    \binom{\gamma d}{i}\cdot\frac{1}{d^{i}}
    \leq
    \left(\frac{e\gamma d}{i}\right)^i\cdot\frac{1}{d^{i}}
    =
    \left(\frac{e\gamma}{i}\right)^i.
\end{align*}
Similarly, for $\delta_1$
we have
\[
    \Pr{\delta_1\geq i}
    =
    \Pr{\Bin(\alpha n,1/n) \geq i}
    \leq
    \left(\frac{e\alpha}{i}\right)^i
    <
    \left(\frac{e\gamma}{i}\right)^i.
\]
Let
$
    p_i=
    \left(\frac{e\gamma}{i}\right)^i.
$
It follows from the above that for any $\ell\geq 1$,
\begin{align}
    \label{eq:prztell}
    \Pr{|\{t\in\{1,\ldots,k\} \colon  \delta_t > i\}| \geq \ell }
    &\leq
    \Pr{\Bin(k,p_i) \geq \ell}
    \leq
    \binom{k}{\ell}\cdot p_i^{\ell}
    \leq
    \left(\frac{e kp_i}{\ell}\right)^\ell.
\end{align}
For $\ell \geq \beta k/i$ and $i\geq e^2\gamma$,
\begin{align*}
    \left(\frac{e kp_i}{\ell}\right)^\ell
    &\leq
    \left(\frac{e k(e\gamma/i)^i}{\beta k/i}\right)^\ell,
    \qquad\text{by } p_i=\left(\frac{e\gamma}{i}\right)^i
    \text{ and } \ell \geq \beta k/i
    \\&
    =
    \left(\frac{e^2 \gamma}{\beta}\cdot \left(\frac{e\gamma}{i}\right)^{i-1}\right)^\ell
    \leq
    \left(\frac{e\gamma}{i}\right)^{(i-1)\ell},
    \qquad\text{by } \beta\geq e^2\gamma
    \\&
    \leq
    \left(\frac{e\gamma}{i}\right)^{(1-1/i)\beta k},
    \qquad\text{by } \ell \geq \beta k/i
    \\&
    \leq
    \left(\frac{1}{e}\right)^{(1-1/e^2)\beta k},
    \qquad\text{by } i\geq e^2\gamma\geq e^2
    \\&
    \leq
    2^{-\beta k}.
\end{align*}
Substituting that to~\eqref{eq:prztell} completes the proof of Lemma~\ref{lem:distributionzt}.
\end{proof}

We proceed now to the proof of the main claim.
For any $A_k\in \AA_k$, and for $\eta = (\theta_0,g_1,\theta_1,\ldots,\theta_k)$ the labeled canonical walk
with semi-concise descriptor $A_k$, let $\EE_{A_k}$ denote the event:
\[
    |\{t\in\{1,\ldots, k\} \colon  2^{\ell-1}< \delta_t \leq 2^{\ell}\}| \leq k_{2^\ell},
    \
    \text{for all}\,
    \ell \in \{\log b+1,\ldots,\log (\alpha n)\}
    .
\]
Applying Lemma~\ref{lem:distributionzt}, for $i=2^{\ell-1}$ and $\beta = b/2$, for each $\ell \in \{\log b+1,\ldots,\log (\alpha n)\}$, and then using a union bound, we obtain
\[
    \Pr{\EE_{A_k}}
    \geq
    1 - 2^{-bk/2} \log (\alpha n).
\]
By another union bound and Lemma~\ref{lem:cardBk},
\begin{align}
    \label{eq:bigcupEEAk}
    \Pr{\bigcap_{A_k\in \BB_k}\EE_{A_k}}
    &\geq
    1- |\BB_k|\cdot 2^{-bk/2}\log(\alpha n)
    \geq
    1- (4b)^{2k}\cdot 2^{-bk/2}\log(\alpha n)
    \notag\\&
    \geq
    1 - 2^{-bk/4}  \log(\alpha n),
\end{align}
where the last inequality holds if $b\geq 64$.
Next we show that event $\bigcap_{A_k\in \BB_k}\EE_{A_k}$ implies that every labeled canonical walk $\eta$ has a concise descriptor $A_k\in \BB_k$.
From this and~\eqref{eq:bigcupEEAk}, the lemma follows.


Fix a realization of \tvisit conditioned on the event $\bigcap_{A_k\in \BB_k}\EE_{A_k}$.
Suppose, for contradiction, that there is some labeled canonical walk $\eta' = (\theta_0',g_1',\theta_1',\ldots,g_k',\theta_k)$ that does \emph{not} have a concise descriptor.
Let $\eta = (\theta_0,g_1,\theta_1,\ldots,g_k,\theta_k)$ be a labeled canonical walk that \emph{does} have a concise descriptor $A_k\in \BB_k$, and shares a \emph{maximal common prefix} with $\eta'$.
Consider the first element where $\eta'$ and $\eta$ are different.
We first argue that this element is not a vertex:
Suppose, for contradiction, that $(\theta_0',\ldots,g_i')=(\theta_0,\ldots,g_i)$ and $\theta'_i\neq \theta_i$, for some $0\leq i\leq k$.
Then $i\neq 0$, as $\theta_0'=s=\theta_0$.
Moreover, if $i>0$, then by definition, $(\theta_0',\ldots,g_i')=(\theta_0,\ldots,g_i)$ implies $\theta'_i = \theta_i$, contradicting our assumption.
Thus, the first element where $\eta'$ and $\eta$ are different must be an agent.
Suppose
$(\theta_0',g_1',\ldots,\theta_{i-1}')=(\theta_0,g_1,\ldots,\theta_{i-1})$ and $g'_{i}\neq g_{i}$, for some $1\leq i\leq k$.
Then, by the maximal prefix assumption, the labeled canonical walk $(\theta_0,\ldots,\theta_{i-1},g'_i, \theta'_i,\bot,\theta'_i,\bot,\ldots,\bot,\theta'_i)$, which stays put at vertex $\theta'_i$ in rounds $i+1$ up to $k$, has no concise descriptor.
This can only be true if
$
    |\{t\in\{1,\ldots,i-1\} \colon  2^{\ell-1}< \delta_t \leq 2^{\ell}\}| > k_{2^\ell},
$
for some $\ell \in \{\log b+1,\ldots,\log n\}$.
But this contradicts event $\EE_{A_k}$.
Therefore, there exists no labeled canonical walk $\eta'$ of length $k$ such that  $\eta'$ has no concise descriptor.

\subsection{Upper Bound on Congestion}

The next lemma gives un upper bound on the congestion of a single canonical walk of length $k$.

\begin{lemma}
    \label{lem:congestionbound}
    Fix any $A_k\in\BB_k$, and let $\eta$
    be the labeled canonical walk with concise descriptor $A_k$  in \tvisit.
    Then, for any $\beta \geq 2e\gamma+1$,
    $
        \Pr{Q(\eta) \leq \beta k} \geq 1- 2^{-(\beta-1) k}.
    $
\end{lemma}
\begin{proof}
Let $\eta = (\theta_0,g_1,\theta_1,\ldots,g_k,\theta_k)$.
Then
$
    Q(\eta) = \sum_{1\leq t\leq k} \delta_t,
$
where $\delta_t = |Z_{\theta_{t-1}}(t-1)|$.
By the same reasoning as in the proof of Lemma~\ref{lem:distributionzt}, $Q(\eta)$ is stochastically dominated by $k+\sum_{1\leq t\leq k} B_t$, where $B_1,\ldots,B_k$  are independent binomial random variables, such that $B_1\sim \Bin(\gamma n,1/n)$ and, for $t>1$, $B_t\sim \Bin(\gamma d,1/d)$.
It follows that
$
    \Exp{Q(\eta) - k} \leq  k\gamma,
$
and
\[
    \Pr{Q(\eta) \geq \beta k}
    =
    \Pr{Q(\eta)-k \geq (\beta-1) k}
    \leq
    2^{-(\beta-1) k},
\]
by a Chernoff bound,
since $(\beta-1)k \geq 2e\cdot\Exp{Q(\eta) - k}$.
\end{proof}

\subsection{Putting the Pieces Together -- Proof of Theorem~\ref{thm:Tpush<Tmsg}}

We consider first the case where $k$ is
at most logarithmic.
In Theorem~\ref{thm:visitx-lower-bound}, we show that $\Tvisit  = \Omega(\log n)$ w.h.p., by arguing that some vertices are not visited by any agent (informed or not) during the first logarithmic number of rounds.
Thus, there is some constant $\epsilon > 0$ such that if $k \leq \epsilon \log n$, $\Pr{\Tvisit \leq k} \leq n^{-\lambda}$.
From this, the theorem's statement follows for $k \leq \epsilon \log n$.
In the rest of the proof, we assume that $k\geq \epsilon \log n$.

We have
$
    \Tpush
    =
    \max_{u\in V} \tau_u,
$
and from Lemma~\ref{lem:tauleqC},
\[
    \Tpush
    \leq
    \max_{u\in V} C_u(t_u)
    .
\]
Since for any fixed realization of \visit and any $u\in V$, $C_u(t)$ is a non-decreasing function of $t$, and since $t_u \leq \Tvisit$, it follows
\[
    \Tpush
    \leq
    \max_{u\in V} C_u(\Tvisit)
    .
\]
By Lemma~\ref{lem:DeqC}, for any $u\in V$, there is a canonical walk $\theta$ of length $t = \Tvisit$ with congestion $Q(\theta) = C_u(\Tvisit)$.
Thus, there is also a \emph{labeled} canonical  walk $\eta$ of length $\Tvisit$ with $Q(\eta) = Q(\theta) = C_u(\Tvisit)$.
It follows
\begin{equation}
    \label{eq:TpushD}
    \Tpush
    \leq
    \max_{\eta\in \HH(\Tvisit)} Q(\eta)
    ,
\end{equation}
where $\HH(t)$ denotes the set of all labeled canonical walks of length $t$ in \visit.

Next we bound $\max_{\eta\in \HH(k)} Q(\eta)$.
Consider \tvisit, and
for any $A_k\in\BB_k$, let $\eta_{A_k}$ be the labeled canonical walk with concise descriptor $A_k$ in \tvisit.
From Lemma~\ref{lem:congestionbound},
for any $A_k\in\BB_k$ and $\beta \geq 2e\gamma+1$,
$
    \Pr{Q(\eta_{A_k}) \leq \beta k} \geq 1- 2^{-(\beta-1) k}.
$
Then
\begin{align*}
    \Pr{\max_{A_k\in \BB_k} Q(\eta_{A_k}) \leq \beta k}
    &\geq
    1 - 2^{-(\beta-1) k}\cdot|\BB_k|
    \geq
    1 - 2^{-(\beta-1) k}\cdot(4b)^{2k},
\end{align*}
by Lemma~\ref{lem:cardBk}.
Choosing constant $\beta$ large enough so that $(\beta-1)/2 \geq 2\log (4b)$, yields
\[
    \Pr{\max_{A_k\in \BB_k} Q(\eta_{A_k}) \leq \beta k}
    \geq
    1 - 2^{-(\beta-1)k/2}.
\]
From Lemma~\ref{lem:consisewhp},
the probability that
all labeled canonical walks of length $k$
have concise descriptors is at least
$1-2^{-bk/4} \log(\alpha n)$, if $b \geq \max\{2\gamma e^2,64\}$.
It follows
\[
    \Pr{\max_{A_k\in \BB_k} Q(\eta_{A_k}) = \max_{\eta\in \HH^\ast(k)} Q(\eta)}
    \geq
    1-2^{-bk/4} \log(\alpha n)
    ,
\]
where $\HH^\ast(t)$ is the set of all labeled canonical walks of length $t$ in \tvisit.
By Lemma~\ref{lem:tweacked}, however, we can couple \visit and \tvisit, by using the same collection of random walks for both, such that the two processes are identical until round $k$ with probability at least $1-kn\cdot 2^{-ad}$.
Thus
\[
    \Pr{ \HH(k) = \HH^\ast(k)}
    \geq
    1-kn\cdot 2^{-\gamma d}.
\]
Combining the last three inequalities above, we obtain
\[
    \Pr{\max_{\eta\in \HH(k)} Q(\eta) \leq \beta k}
    \geq
    1 - 2^{-(\beta-1)k/2}
    -2^{-bk/4} \log(\alpha n)
    - kn\cdot e^{-\gamma d}.
\]
Since $k\geq \epsilon \log n$ and $d\geq \varepsilon \log n$,
for any given constant $\lambda>0$
we can choose constants $\beta, b, \gamma$ large enough such that
\begin{equation}
    \label{eq:allk}
    \Pr{\max_{\eta\in \HH(k)} Q(\eta) \leq \beta k}
    \geq
    1 - n^{-\lambda}.
\end{equation}
From~\eqref{eq:TpushD} and~\eqref{eq:allk}, we obtain
\begin{align*}
    \Pr{\Tpush \leq \beta k}
    &\geq
    \Pr{\max_{\eta\in \HH(\Tvisit)} Q(\eta) \leq \beta k},
    \text{\qquad by ~\eqref{eq:TpushD}}
    \\&
    \geq
    \Pr{\{\Tvisit \leq k\} \cap \left\{\max_{\eta\in \HH(k)} Q(\eta) \leq \beta k\right\}}
    \\&
    \geq
    \Pr{\Tvisit \leq k } - \Pr{\max_{\eta\in \HH(k)} Q(\eta) > \beta k}
    \\&
    \geq
    \Pr{\Tvisit \leq k } - n^{-\lambda},
    \text{\qquad  by ~\eqref{eq:allk}}.
\end{align*}
This completes the proof of Theorem~\ref{thm:Tpush<Tmsg}.

\section{Bounding \texorpdfstring{$\Tvisit$}{Tvisitx} by \texorpdfstring{$\Tpush$}{Tpush} on Regular Graphs}
\label{sec:msg<push}

The following theorem upper bounds the broadcast time of \visit in a regular graph by the broadcast time of \push.

\begin{theorem}
    \label{thm:Tmsg<Tpush}
    For any constants $\alpha,\beta,\lambda > 0$ with $\alpha\cdot \beta$ sufficiently large, there is a constant $c > 0$, such that for any $d$-regular graph $G = (V, E)$ with $|V| = n$ and $d \geq \beta \log n$,  and for any source $s\in V$, the broadcast times of \push and \visit, with $|A| \geq \alpha n$ agents, satisfy
    \[
        \Pr{\Tvisit \leq ck} \geq \Pr{\Tpush \leq k} - n^{-\lambda},
    \]
    for any $k\geq 0$.
    \htodo{$k$ is used in the proofs as the length of path.}
\end{theorem}

From Theorem~\ref{thm:Tmsg<Tpush}, it is immediate that if $\Tpush \leq T$ w.h.p., then $\Tvisit = O(T)$ w.h.p.
Moreover, using Theorem~\ref{thm:Tmsg<Tpush} and the well-known $O(n^2\log n)$ upper bound w.h.p.\ on the cover time for a single random walk on a regular graph, which also applies to $\Tvisit$, one can easily obtain that $\Exp{\Tvisit} = O(\Exp{\Tpush})$.

%
%

\paragraph{Proof Overview Of Theorem~\ref{thm:Tmsg<Tpush}.}

We use a coupling which is  similar to that in the proof of the converse result, stated in Theorem~\ref{thm:Tpush<Tmsg}, but with a twist (which we describe momentarily).
Unlike in the proof of Theorem~\ref{thm:Tpush<Tmsg}, where we essentially consider all possible paths through which information travels, here we focus on the first path by which information reaches each vertex.
Let $P = (u_0=s,u_1,\ldots,u_k=u)$ be such a path for vertex $u$ in \push, where each vertex $u_i$ in the path gets informed by $u_{i-1}$.
Let $\delta_i$ be the number of rounds it takes for $u_{i-1}$ to sample (and inform) $u_{i}$ in \push.
We consider the same path in \visit, and compare $\delta_i$ with the number $D_i$ of rounds until some informed agent moves from $u_{i-1}$ to $u_i$, counting from the round when $u_{i-1}$ becomes informed.
Note that $\sum_i \delta_i$ is precisely the round when $u$ is informed in \push, while  $\sum_i D_i$ is an upper bound on the round when $u$ is informed in \visit.

The coupling from Section~\ref{sec:push<msg} seems suitable for this setup.
Recall, in that coupling we let the list of neighbors that a vertex $u$ samples in \push, be identical to the list of neighbors that informed agents visit in their next step after visiting $u$, in \visit.
The same intuition applies, namely, that on average each vertex is visited by $|A|/n = \Omega(1)$ agents per round, which suggests that $D_i$ should be close to $\delta_i$.
We can even apply a similar trick as in Section~\ref{sec:push<msg} to avoid some dependencies:
In each round, the number of agents in the neighborhood of a vertex is bounded below by $d \cdot |A| /n = \Omega(d)$, w.h.p.
This should imply that the number of agents that visit a vertex in a round is bounded below by a geometric distribution with constant expectation.
Let $\EE$ denote the event that the above $\Omega(d)$ bound holds for all $u$, for polynomially many rounds.

There is, however, a problem with this proof plan.
By fixing path $P$ in advance, to be the first path to inform $u$ in \push, we introduce \emph{dependencies from the future}.
So, when we analyse $D_i$ and $\delta_i$, we must condition on the event that the $i$-prefix of the path we have considered so far will indeed be a prefix of the first path to reach $u$.
These kind of dependencies seem hard to deal with.

We use the following neat idea to overcome this problem.
We only consider the \emph{odd rounds} of \visit in the coupling, i.e., we match the list of neighbors that a vertex $v$ samples in \push (in all rounds), to the list of neighbors that informed agents visit in round $2k+1$ after visiting $u$ in round $2k$, for all $k\geq 0$.
In even rounds, agents take steps \emph{independently} of the coupled \push process.

Under this coupling, we proceed as follows.
We condition on the high probability event $\EE$ defined earlier (formally, we modify \visit to ensure $\EE$ holds).
We then fix \emph{all} random choices in \push, and thus the information path $P$ to $u$.
For each even round of \visit, we have that vertex $u_i$ in $P$ is visited by at least one agent with constant probability, independently of the past and of the fixed choices in future odd rounds.
If indeed some vertex visits $u_i$ in an even round, then in the next round it will visit a vertex dictated by the coupling.
This allows us to show that under this coupling,
$\sum_i D_i \leq c\left(\sum_i \delta_i + \log n\right)$, w.h.p.
We get rid of the $\log n$ term in the final bound, by using that $\Tpush  = \Omega(\log n)$ w.h.p.

\subsection{Coupling Description}
\label{sec:visit<push-coupling}
We use mostly the same notation as in Section~\ref{sec:push<msg-coupling}.
For each vertex $u$, we denote by
$\tau_u$ the round when vertex $u$ gets informed in \push.
For $i \geq 1$, let $\pi_u(i)$ be the $i$th the vertex that $u$ samples (in round $\tau_u + i$).
We denote by $t_u$ the round when vertex $u$ gets informed in \visit.
For an agent $g \in A$ and round $t \geq 0$, let $x_g(t)$ be the vertex that $g$ visits in round $t$.
Let $Z_u(t)$ be the set of agents that visit $u$ in round $t$, i.e.,
$
    Z_u(t) = \{g\in A\colon x_g(t)=u\}.
$

The next definition differs from the corresponding one in Section~\ref{sec:push<msg-coupling}, as it distinguishes between even and odd rounds.
Fix a vertex $u \in V$, and
consider all visits to $u$ in \emph{even} rounds $t\geq t_u$, in chronological order, ordering visits in the same round with respect to a predefined total order over all agents.
We call these visits \emph{even visits} to vertex $u$.
For each $i \geq 1$, consider the agent $g$ that performs the $i$th even visit
and let $p_u^{odd}(i)$ be the vertex that $g$ visits in the next (\emph{odd}) round.
Formally, let
\[
    W_u^{even} =\{(t,g)\colon t\geq t_u, t\in \IN_{even}, x_g(t) = u\},
\]
where $\IN_{even}$ is the set of non-negative even integers.
Order the elements of $W_u^{even}$ such that $(t,g)<(t',g')$ if $t<t'$, or $t=t'$ and $g<g'$.
If $(t,g)$ is the $i$th smallest element in $W_u^{even}$, then $p_u^{odd}(i) = x_g(t+1)$.

\paragraph{Coupling.}
We couple processes \push and \visit by setting $\pi_u(i) = p_u^{odd}(i)$.
Formally, let $\{w_u(i)\}_{u\in V,i\geq 1}$, be a collection of independent random variables
each taking a uniformly random value from the set $\Gamma(u)$ of $u$'s neighbors in $G$.
For all $u\in V$ and $i\geq 1$, we set
\[
    \pi_u(i) = p_u^{odd}(i) = w_u(i)
    .
\]

\subsection{Lower Bound on Agents and Re-Tweaked Visit-Exchange}
\label{sec:tweaked2}

We will use the following simple lower bound on the number of agents visiting a given set of vertices $S$ in a round of \visit.
The proof is almost the same as its counterpart Lemma~\ref{lem:msgdensityub}.

\begin{lemma}
    \label{lem:msgdensitylb}
    For any $S\subseteq V$ and $t\geq0$,
    \[
        \Pr{\sum_{v\in S}|Z_v(t)|
        \geq
        |A|\cdot|S|/(2n)}
        \geq
        1-e^{-|A|\cdot|S|/(8n)}.
    \]
\end{lemma}

\begin{proof}
Since each agent's walk starts from the stationary distribution and $G$ is a regular graph,
we have that for any given agent $g \in A$ and round $t \geq 0$,
$\Pr{x_g(t)\in S} = |S|/n$.
Therefore the expected number of agents visiting $S$ in round $t$ is
\[
    \Exp{|\{g\in A \colon x_g(t)\in S\}|} = |A|\cdot |S|/n.
\]
By the independence of the walks, we can use a standard Chernoff bound to show that $|\{g\in A \colon x_g(t)\in S\}|\geq |A|\cdot |S|/(2n)$,
with probability at least $1-e^{-|A|\cdot|S|/(8n)}$.
\end{proof}

\paragraph{Re-Tweaked Visit-Exchange Process.}

Similar to the analysis in Section~\ref{sec:push<msg}, it is convenient to work with a slightly modified version of \visit.
We call the new process \rvisit and is identical to \visit except for the following modification.
If in some odd round $t\geq 0$, there is a vertex $u\in V$ for which the next condition is \emph{not} true,
\begin{equation}
    \label{eq:alphaConditionRT}
    \sum_{v \in \Gamma(u)} |Z_v(t)| \geq \frac{|A|}{2n}\cdot d
\end{equation}
then before round $t + 1$, we add a minimal set of new agents to the graph such that the above condition holds for all vertices $u$.
An agent $g$ added to vertex $u$ adopts the state (informed or non-informed) of $u$ at the end of round $t$.

Recall that $|A| \geq \alpha n$.
The next lemma allows us to consider the \rvisit process in the rest of the proof, and argue that the results also hold for \visit.
\begin{lemma}
    \label{lem:retweacked}

    The probability that Eq.\eqref{eq:alphaConditionRT} holds simultaneously  for all $u\in V$ and $0 \leq t < k$ is at least $1-kn\cdot 2^{-\alpha d/8}$.
\end{lemma}
\begin{proof}
For each $u \in V$, if we set $S = \Gamma(u)$,
then Lemma~\ref{lem:msgdensitylb} implies that the condition \eqref{eq:alphaConditionRT}
holds with probability at least $1 - e^{|A|\cdot |S| / (8n)} \geq 1 - e^{\alpha d / 8}$.
The claim in the lemma follows after applying union bound for each $0 \leq t < k$ and each $u \in V$.
\end{proof}


\subsection{Proof of Theorem~\ref{thm:Tmsg<Tpush}}
We first compare the times until a given vertex $u$ gets informed in \push and  in \rvisit.

\begin{lemma}
    \label{lem:retweak<push}
    The coupling described in Section~\ref{sec:visit<push-coupling}, when applied to \push and \rvisit, yields the following property.
    For any constant $\gamma > 0$, there is a constant $c > 0$ such that for any $u\in V$,
    \[
        \Pr{t_u' \geq c(\tau_u + \log n)} \leq n^{-\gamma},
    \]
    where $\tau_u$ and $t'_u$ are the rounds when $u$ is informed in the coupled processes \push and in \rvisit, respectively.
\end{lemma}

\begin{proof}
In this proof, we will use the same notation for \rvisit as those defined for \visit.
(We used $t'_u$ instead of $t_u$ in the lemma's statement to avoid confusion when we apply the lemma, but in the proof there is no such fear, because only \rvisit is used.)

As described in the proof overview, we consider a path from the source $s$ to vertex $u$ that \push uses to inform $u$, and count the number of rounds \visit takes to traverse the same path.
First, we consider a single edge $(v, w)$ such that $w$ is informed by $v$ in a realization of \push that we fix.
We also fix the first $t_v$ rounds of \rvisit, i.e., until $v$ becomes informed.
Let $\delta_{v,w} = \tau_w - \tau_v$ be the number of rounds of \push that it takes to inform $w$ counting from when $v$ gets informed.
Similarly, we define $D_{v,w} = t_w - t_v$ for \rvisit.
We will bound $D_{v, w}$ in terms of $\delta_{v, w}$.


Recall that we have defined a natural total order over the set $W_v^{even}$ of even visits to vertex $v$.
For $j \geq 1$, let $(t, g)$ be the $j$th element of $W_v^{even}$ in that order.
By the coupling, at the odd round $t + 1$, agent $g$ will move to the neighbor of $v$ that is sampled by \push in round $\pi_v(j) = \tau_v + j$.
In particular, since $\pi_v(j) = w$ for $j = \delta_{v,w}$, vertex $w$ gets informed after $\delta_{v,w}$ even visits to $v$ in \rvisit (possibly earlier).

Formally, let $B_v^{(j)}$ be the number of \rvisit rounds between even visits $j - 1$ and $j$ (when $j = 1$, $B_v^{(j)}$ is the number of rounds until the first even visit  since $t_v$).
$B_v^{(j)}$ can be $0$, if two agents visit $v$ at the same even round.
With this definition,
\[
    D_{v,w} \leq \sum_{j=1}^{\delta_{v,w}} B_v^{(j)}.
\]
%
By condition \eqref{eq:alphaConditionRT} and assumption $|A| \geq \alpha\cdot n$, there are at least $\alpha\cdot d/2$ agents in the neighborhood of $v$ at any round of \rvisit.
Let $p = 1 - e^{-\alpha / 2}$ and
recall that, for an even $t > 0$, the agents move independently from \push,
and therefore, some agent visits $v$ in round $t$ with probability at least $1 - (1 - 1 / d)^{\alpha d / 2} \geq p$.
For $t = 0$, when agents are placed according to the stationary distribution,
some agent is placed at $v$ with probability $1 - (1 - 1/n)^{\alpha n} \geq 1 - e^{-\alpha} \geq p$.
It follows that the number of rounds between two even visits to $v$, namely $B_v^{(j)}$ for $1 \leq j \leq \delta_{v,w}$, is stochastically dominated by $2\cdot F_v^{(j)}$, where $\{ F_v^{(j)} \}_{j \geq 1}$ is a collection of independent geometric random variables with success probability $p$.
The coefficient $2$ appears because we have to take into account both odd and even rounds.
In other words, for any $b \geq 0$ and $1 \leq j \leq \delta_{v,w}$,
\[
    \Pr{B_v^{(j)} \leq b \mid  B_v^{(1)}, \dots, B_v^{(j-1)} } \geq \Pr{2\cdot F_v^{(j)} \leq b}.
\]
Using Lemma \ref{lem:stochastic-sum}, we get that, given $v$ is informed, $D_{v,w}$ is stochastically dominated by $2\cdot \sum_{j=1}^{\delta_{v,w}}F_v^{(j)}$:
\[
    \Pr{D_{v,w} \leq b \mid t_v} \geq \Pr{ \sum_{j=1}^{\delta_{v, w}} B_v^{(j)} \leq b \mid t_v } \geq \Pr{2\cdot \sum_{j=1}^{\delta_{v,w}}F_v^{(j)} \leq b }.
\]

We apply the above result to all edges on the path from $s$ to $u$ through which \push informed~$u$.
Let $P_u = (s = u_0, u_1, \dots, u_k = u)$ be a path in $G$ such that, in \push, $u_i$ is informed from $u_{i-1}$, for all $1 \leq i \leq k$.
By definition of $\tau_u$, $u_{i - 1}$ samples its neighbor $u_i$ at round $\tau_{u_i}$.
Define $\delta_i = \tau_{u_i} - \tau_{u_{i-1}}$ and $D_i = t_{u_i} - t_{u_{i-1}}$ for $1 \leq i \leq k$.
From our result above for a single edge it follows that
\[
    \Pr{D_i \leq b \mid D_1,\dots,D_{i-1}} \geq \Pr{2\cdot \sum_{j=1}^{\delta_i} F_{u_i}^{(j)}\leq b}.
\]
By Lemma~\ref{lem:stochastic-sum} and the fact that $t_u = t_{u_k} = \sum_{i=1}^k D_i$, we have that
$t_u$ is stochastically dominated by $2F = 2\cdot \sum_{i=1}^k \sum_{j=1}^{\delta_i} F_{u_{i-1}}^{(j)}$, i.e., for any $b \geq 0$,
\[
    \Pr{ t_u \leq b } \geq \Pr{ 2F \leq b }.
\]
The random variable $F$ is a sum of exactly $\tau_k$ independent and identical geometrically distributed random variables, hence, $\E{F} = \tau_k / p$.
Thus, for any constant $c \geq 4 / p$, by Lemma~\ref{lem:geom-chernoff},
%
%
\begin{align*}
    \Pr{t_{u} \geq c(\tau_{u} + \log n)}
    &\leq \Pr{F \geq \frac{c}{2}(\tau_{u} + \log n)} \\
        &\leq
        \exp\left(-\frac{c(\tau_{u} + \log n)\cdot p}{16}\right) \\
        &\leq n^{-cp/16},
\end{align*}
Choosing $c$ large enough so that $cp / 16 \geq \gamma$, completes the proof.
\end{proof}

We can now complete the proof of our main result.
%
    Recall that $\tau_u, t_u$ and $t_u'$ are the rounds when vertex $u$ gets informed in \push, \visit, and \rvisit, respectively.
    From Lemma \ref{lem:retweak<push}, and a union bound over all vertices, we obtain that for any constant $\gamma > 0$, there is a constant $c > 0$ such that
    \[
        \Pr{\forall\,u\in V\colon t_u' \leq c(\tau_u + \log n)} \geq 1-n\cdot n^{-\gamma}.
    \]
    Thus,
    \[
        \Pr{\max_{u\in V} t_u' \leq c\left(\max_{u\in V} \tau_u + \log n\right)} \geq 1-n\cdot n^{-\gamma}.
    \]
    It follows that for any $k\geq 0$,
    \begin{align*}
        \Pr{\max_{u\in V} t_u' \leq c\left(k + \log n\right)}
        &\geq
        \Pr{\max_{u\in V} t_u' \leq c\left(\max_{u\in V} \tau_u + \log n\right) \cap \max_{u\in V} \tau_u \leq k}
        \\&
        \geq
        \Pr{\max_{u\in V} \tau_u \leq k} -n\cdot n^{-\gamma}.
    \end{align*}
    From Lemma~\ref{lem:retweacked}, it follows
    \begin{align*}
        \Pr{\max_{u\in V} t_u' \leq c\left(k + \log n\right)}
        -
        \Pr{\max_{u\in V} t_u \leq c\left(k + \log n\right)}
        \leq
        c (k + \log n )\cdot n\cdot e^{-\alpha d/8}.
    \end{align*}
    Combining the last two inequalities above we obtain
    \begin{align*}
        \Pr{\max_{u\in V} t_u \leq c\left(k + \log n\right)}
        \geq
        \Pr{\max_{u\in V} \tau_u \leq k} -n\cdot n^{-\gamma}
            -c (k + \log n )\cdot n\cdot e^{-\alpha d/8}.
    \end{align*}
    Substituting $\Tvisit = \max_{u\in V} t_u$ and $\Tpush = \max_{u\in V} \tau_u$, and using
    $d \geq \beta \log n$,
    yields
    \begin{align*}
        \Pr{\Tvisit \leq c\left(k + \log n\right)}
        \geq
        \Pr{\Tpush \leq k} - n^{-\gamma+1}
            -c(k + \log n)\cdot n^{1 - \alpha\beta / 8}.
    \end{align*}
    This implies the theorem for $\log n \leq k \leq \mathrm{poly}(n)$.
    For larger $k$, the theorem follows from the known polynomial upper bound on the cover time on regular graphs.
    For smaller $k$, it follows from the fact that $\Tpush = \Omega(\log n)$, w.h.p.

\section{Bounding \texorpdfstring{$\Tvisit$}{Tvisit} by \texorpdfstring{$\Tmeet$}{Tmeetx} on Regular Graphs}
\label{sec:visit<meet}

The next theorem bounds the broadcast time of \visit on a regular graph by the broadcast time of \meet.

\begin{theorem}
    \label{thm:Tvisit<Tmeet}
    For any constants $\alpha,\beta,\lambda > 0$ with $\alpha\cdot \beta$ sufficiently large, there is a constant $c > 0$, such that for any $d$-regular graph $G = (V, E)$ with $|V| = n$ and $d \geq \beta \ln n$, and any source $s\in V$, the broadcast times of \visit and \meet, both with $|A| \geq \alpha n$ agents, satisfy
    \[
        \Pr{\Tvisit \leq k + c\ln n} \geq \Pr{\Tmeet \leq k} - n^{-\lambda},
    \]
    for any $k\geq 0$.
\end{theorem}

\begin{proof}
    Let $\Rvisit$ be the number of rounds until all \emph{agents} are informed in \visit.
    Under the natural coupling of \visit and \meet,
    that uses the same random walks for both processes,
    it is immediate that
    \begin{align}
        \label{eq:rvisit-vs-tmeet}
        \Pr{ \Rvisit \leq k }
            \geq \Pr{ \Tmeet \leq k }.
    \end{align}
    Let $\ell = c\ln n$ for constant $c$ to be determined later.
    Next we show that in $\ell$ rounds of \visit, every vertex is visited by at least one agent, with probability at least $1 - n^{-\lambda}$.
    For that we consider the process \rvisit from Section \ref{sec:tweaked2}, which ensures that for every vertex $u \in V$ and round $t \geq 0$,
    \[
        \sum_{v \in \Gamma(u)} |Z_v(t)| \geq |A| \cdot d / (2n) \geq \alpha d / 2,
    \]
    where $Z_v(t)$ is the set of agents visiting $v$ in round $t$.
    
    Fix a vertex $u$.
    In any round $t \in R_k = \{ k + 1, \dots, k + \ell\}$ of \rvisit,
    the probability that no agent visits $u$ in that round is at most
    $(1 - 1 / d)^{\alpha d / 2} \leq e^{-\alpha / 2}$,
    since the neighborhood of $u$ contains at least $\alpha d / 2$ agents
    before round $t$.
    This holds for every round $t$ independently,
    hence $u$ is visited by some agent in rounds $R_k$
    with probability at least $1 - e^{-\alpha \ell / 2}$.
    By a union bound, with probability at least $1 - n \cdot e^{-\alpha \ell / 2 }$,
    every vertex $u$ is visited by some agent in rounds $R_k$.
    By Lemma \ref{lem:retweacked}, \rvisit and \visit are identical in the first $(k + \ell)$ rounds of their execution with probability at least
    $1 - (k + \ell) n \cdot 2^{-\alpha d / 8}$.
    The last two statements together imply that
    \begin{align*}
        \Pr{ \Tvisit \leq k + \ell }
            &\geq \Pr{ \Tvisit \leq k + \ell \mid \Rvisit \leq k } \cdot
                 \Pr{ \Rvisit \leq k } \\
            &\geq \left(1 - (k + \ell) n \cdot 2^{-\alpha d / 8} - n \cdot e^{-\alpha \ell / 2 } \right) \cdot \Pr{ \Rvisit \leq k }\\
            &\geq \Pr{ \Rvisit \leq k } - (k + l) n^{1 - \alpha \beta / 16} - n^{1 - \alpha c / 2}.
    \end{align*}
    Together with~\eqref{eq:rvisit-vs-tmeet}, this implies the theorem for $\mathrm{poly}(\log n) \leq k \leq \mathrm{poly}(n)$,
    since we can take $\alpha \cdot \beta$ and $c$ sufficiently large,
    depending on $\lambda$.
    For smaller $k$, the theorem follows from
    the fact that $\Tmeet = \Omega(\log n)$ w.h.p.~(Theorem \ref{thm:meetx-lower-bound}).
    For larger $k$, it follows from the fact that $\Tvisit, \Tmeet \leq \mathrm{poly}(n)$ w.h.p., by a known polynomial upper bound on the cover time of a random walk in a graph.
\end{proof}

\section{Logarithmic Lower Bounds for  \texorpdfstring{$\Tvisit$}{Tvisit} \& \texorpdfstring{$\Tmeet$}{Tmeet} on Regular Graphs}

\begin{theorem}
    \label{thm:visitx-lower-bound}
    For any $d$-regular graph $G = (V, E)$ with $|V| = n$ and $d = \Omega(\log n)$, and any source vertex $s \in V$, the broadcast time of \visit with $|A| = O(n)$ agents is $\Omega(\log n)$ w.h.p.
\end{theorem}
\begin{proof}
We argue that w.h.p.\ some vertices are not visited by any agent (informed or not)
during the first logarithmic number of rounds of \visit.
We only count the visits starting from round $1$, since the initial placement of agents cannot inform any vertex.
The formal argument follows next.

For a sufficiently large constant $\gamma > 0$, that will be fixed later,
we consider the process \tvisit defined in Section~\ref{sec:tweaked}.
Recall that in \tvisit, for every vertex $u \in V$,
\[
    \sum_{v \in \Gamma(u)} |Z_v(t)| \leq \gamma \cdot d,
\]
where $Z_v(t)$ is the set of agents that visit $v$ in round $t$.
In the rest of the proof we use \tvisit and use the fact that it is equivalent to
\visit w.h.p.~for the first logarithmic rounds of the process.

Let $\HH_t$ represent all random choices of \tvisit up to (and including) round $t$,
and let $U_t$ be the set of vertices that have not been visited by any agent (either informed, or not) in any round up to $t$.
Denote the event that $|U_t| \geq |U_{t-1}| \cdot 4^{-\gamma} / 2$ by $\AA_t$.
We will show that for any $t \geq 1$,
\begin{align}
    \label{eq:u-shrink-lbd}
    \Pr{ \AA_t \;\middle|\; \HH_{t-1}; |U_{t-1}| \geq \log^2 n }
                = 1 - n^{-\omega(1)}
                \geq 1 - n^{-\lambda - 1},
\end{align}
for any constant $\lambda > 0$.
By the definition of \tvisit,
for each $u \in U_{t-1}$, the total number of agents in $\Gamma(u)$ before round $t$ is at most $\gamma d$.
Each of these agents visits $u$ in round $t$ with probability $1 / d$, independently from one another.
Let $X_u$ be the indicator random variable that $u \in U_t$.
Then, for $u \in U_{t-1}$,
\[
    \Pr{ X_u = 1 \mid \HH_{t-1} }
        \geq (1 - 1/d)^{\gamma d} \geq 4^{-\gamma},
\]
which implies that
\[
    \E{ |U_t| \mid \HH_{t - 1} }
        = \E{ \sum_{u \in U_{t-1}} X_u \;\middle|\; \HH_{t-1} }
        \geq |U_{t-1}| \cdot 4^{-\gamma}.
\]
We observe that, conditioned on the history $\HH_{t-1}$, the random variables $X_u$ are negatively associated~\cite[Example~3.1]{Dubhashi_Panconesi_2009}.
Thus, we can apply standard Chernoff bounds on their sum to obtain
\[
    \Pr{ |U_t| \geq |U_{t-1}|\cdot 4^{-\gamma} / 2 \;\middle|\; \HH_{t-1} }
        \geq 1 - \expp{ |U_{t-1}| \cdot 4^{-\gamma} / 8 },
\]
which implies \eqref{eq:u-shrink-lbd}.

Let $\kappa = \floor { \log_{2 \cdot 4^{\gamma}}(n / \log^2n) }$,
and for $t \in \{ 1, \dots, \kappa \}$,
define $\XX_t = \bigcap_{1 \leq t' \leq t} \AA_{t'}$.
We prove that $\Pr{ \XX_t } \geq 1 - t\cdot n^{-\lambda - 1}$ by induction.
The $t = 1$ case is exactly the statement of inequality \eqref{eq:u-shrink-lbd} since $|U_0| = |V| = n$.
For $t > 1$,
\begin{align*}
    \Pr{ \XX_t }
        &\geq \Pr{ \AA_t \mid \XX_{t - 1} } \cdot \Pr{ \XX_{t - 1} } \\
        &\geq \left(1 - n^{-\lambda - 1}\right) \cdot \Pr{ \XX_{t - 1} },
        \quad \text{by \eqref{eq:u-shrink-lbd} since $\XX_{t-1}$ implies $|U_{t-1}| \geq \log^2n$,} \\
        &\geq \left(1 - n^{-\lambda - 1}\right) \cdot \left(1 - (t - 1)\cdot n^{-\lambda - 1}\right),
        \quad \text{by the inductive hypothesis,}\\
        &\geq 1 - t \cdot n^{-\lambda - 1}.
\end{align*}
Observe that $\XX_\kappa$ implies that there are at least $\log^2n$ vertices
that have not been visited by any agent, and thus
at least $\log^2n - 1$ vertices that are uninformed (the other one may be the source).
Therefore, with probability at least $1 - \kappa\cdot n^{-\lambda - 1}$,
there is an uninformed vertex in \tvisit after round $\kappa$.
By Lemma \ref{lem:tweacked},
\tvisit and \visit are identical in the first $\kappa$ rounds of their execution,
with probability at least $1 - \kappa n 2^{-\gamma d}$.
Combining the two statements,
we get that there is an uninformed vertex in \visit after round $\kappa$,
with probability at least
$1 - \kappa \cdot n^{-\lambda - 1} - \kappa \cdot n 2^{-\gamma d}$.
By choosing a sufficiently large $\gamma$ and using the fact that $d = \Omega(\log n)$,
we can make this probability to be at least $1 - n^{-\lambda}$,
while $\kappa = \Omega(\log n)$,
completing the proof.
\end{proof}

\begin{theorem}
    \label{thm:meetx-lower-bound}
    For any $d$-regular graph $G = (V, E)$ with $|V| = n$ and $d = \Omega(\log n)$, and any source vertex $s \in V$, the broadcast time of \meet with $|A| = O(n)$ agents is $\Omega(\log n)$ w.h.p.

\end{theorem}
\begin{proof}
The proof follows the same line of logic as the proof of Theorem ~\ref{thm:visitx-lower-bound}.
We show that, w.h.p., there is an agent that has not started its walk at the source, and that has not met any other agent (informed or not) in the first logarithmic number of rounds of \meet.

First observe that we can consider a tweaked process \tmeet, which has the same modification as \tvisit in Theorem~\ref{thm:visitx-lower-bound} that ensures that
the neighborhood of every vertex contains at most $O(d)$ agents at any round.
Recall that $\HH_t$ is the history of \tmeet until round $t$.
Let $S_t$ be the set of agents that have not met another agent in the first $t$ rounds, and
let $\AA_{t}$ be the event that
$|S_t|$ is a constant fraction of $|S_{t-1}|$.
The next inequality, which is analogous to~\eqref{eq:u-shrink-lbd}, is the key step of the proof and is proved next:
\begin{align}
    \label{eq:s-shrink-lbd}
    \Pr{ \AA_t \;\middle|\; \HH_{t-1} ; |S_{t-1}| \geq \log^2 n } \geq 1 - n^{-\lambda - 1}.
\end{align}
For every agent $g \in S_{t-1}$, consider the vertex $u = x_g(t)$ that $g$ visits in round $t$.
With constant probability no agent other than $g$ visits $u$ in round $t$,
therefore, there is a constant $\beta$ such that
$\E{ |S_t| \mid \HH_{t-1} } \geq \beta |S_{t-1}|$.
Unlike in Theorem~\ref{thm:visitx-lower-bound}, we do not have negative association of the events that agents in $S_{t-1}$ are also in $S_t$,
and therefore cannot use Chernoff bound directly.

Instead, we split round $t$ into two sub-rounds:
In the first sub-round, only the agents in $S_{t-1}$ make a step, and in the second one all other agents.
Consider the set $S'_t$, which contains agents $g \in S_{t-1}$ that do not meet another agent from $S_{t-1}$ in the first sub-round.
We have that
$\E{|S_t'| \mid \HH_{t-1}} \geq \E{|S_t| \mid \HH_{t-1} } \geq \beta |S_{t-1}|$.
Additionally,
$|S_t'|$ is a function of the independent steps taken by the agents in $S_{t-1}$,
and changing the step of one of them changes $|S_t'|$ by at most $2$.
It implies that, by the Method of Bounded Difference~\cite[Corollary 5.2]{Dubhashi_Panconesi_2009},
\[
    \Pr{ |S_t'| \geq \beta |S_{t-1}| / 2 \mid \HH_{t-1} } \geq 1 - e^{-\Omega(|S_{t-1}|)}.
\]
Consider the set of vertices $L_t$ where agents in $S_t'$ are located after the first sub-round.
We can now use the negative association argument from Theorem~\ref{thm:visitx-lower-bound}
to show that, with probability at least $1 - e^{-\Omega(|L_t|)} = 1 - e^{-\Omega(|S_{t-1}|)}$,
a constant fraction of vertices in $L_t$ do not receive any agent in the second sub-round.
Hence, with the same probability, a constant fraction of agents in $S_t'$ do not meet a new agent.
Combining the above arguments, we prove~\eqref{eq:s-shrink-lbd}.

Applying~\eqref{eq:s-shrink-lbd} for $\kappa = \Omega(\log n)$ rounds,
we get that, w.h.p., at least $\log^2n$ agents have not met any other agent after the first $\kappa$ rounds.
Of these agents, at most $O(\log n)$ get informed in rounds $0$, w.h.p.
This follows from a standard bound on the largest bin in the balls-and-bins problem.
Additionally, at most one such agent could be the first one to visit $s$, while $s$ still contains the information.
Therefore, the broadcast time of \tmeet and thus also \meet is at least $\kappa = \Omega(\log n)$.
\end{proof}
%

\section{Open Problems}
\label{sec:open}

This work is the first systematic and thorough comparison of the running times of the standard \push and \ppull rumor spreading protocols with some very natural agent-based alternatives.
Several open problems remain.
The most obvious question to ask is whether our results for regular graphs hold also when the graph degree is sub-logarithmic.
Another question is whether there are graphs where \meet is slower than \visit by more than logarithmic factors.
In this paper we assumed a linear number of agents.
It would be interesting to study the performance of the protocols when a sub-linear number of agents is available.

The main attractive properties of standard rumor spreading protocols are simplicity, scalability, and robustness to failures~\cite{Feige1990}.
Arguably, \visit and \meet share the first two properties, but probably not the robustness property. 
In particular, it seems that faulty nodes or links can result in agents getting lost.
It would be interesting to explore fault tolerant variants of these protocols.
For example, it seems likely that the protocols could tolerate some number of lost agents, if a \emph{dynamic} set of agents were used, where agents age with time and die, while new agents are born at a proportional rate.

\section{Acknowledgments}

We would like to thank Thomas Sauerwald and Nicol\'{a}s Rivera for helpful discussions.
This research was undertaken, in part, thanks to funding from
the ANR Project PAMELA (ANR-16-CE23-0016-01),
the NSF Award Numbers CCF-1461559, CCF-0939370 and CCF-18107,
the Gates Cambridge Scholarship programme,
and the ERC grant DYNAMIC MARCH.

\bibliographystyle{abbrv}
\bibliography{frogwalk}

\appendix
\section*{APPENDIX}

\section{Concentration Bounds}

Below we state some standard bounds we use in our analysis.

\begin{theorem}[Chernoff bounds, {\cite[Theorems 4.2, 4.3]{Mitzenmacher_Upfal_2017}}]
    \label{thm:chernoff}

    Let $X_1, X_2, \dots, X_n$ be independent 0/1 random variables. Let $X = \sum_{i=1}^n X_i$ and $\mu = \E{X}$.
    Then,
    \begin{center}
        \begin{enumerate}[(a)]
            \item $\Pr{X \geq (1 + \delta)\cdot \mu } \leq \exp\left(-\frac{\mu \cdot \delta^2}{3}\right)$, for $0 < \delta \leq 1$.
            \item $\Pr{X \geq \beta\mu} \leq 2^{-\beta\mu}$, for $\beta \geq 2e$.
            \item $\Pr{X \leq (1 - \delta)\cdot \mu } \leq \exp\left(-\frac{\mu\cdot \delta^2}{2}\right)$, for $0 < \delta < 1$.
        \end{enumerate}
    \end{center}
\end{theorem}

\begin{lemma}
    \label{lem:geom-chernoff}
    Let $F_1, \dots, F_n$ be independent and identical geometrically distributed random variables with parameter $p$, i.e., for any integer $k \geq 1$,
    $ \Pr{F_i = k} = (1 - p)^{k-1}p. $
    Let $F = \sum_{i=1}^n F_i$ and $\mu = \E{F}$. Then for any
    $k\geq 2\mu$,
    \[ \Pr{F \geq k} \leq \exp\left(-\frac{kp}{8}\right).\]
\end{lemma}
\begin{proof}
    We define a coupling between random variables $(F_i)_{i=1}^n$ and a sequence of Bernoulli trials $(X_j)_{j=1}^{\infty}$ with parameter $p$.
    Let $j_0 = 0$ and for $i\geq 1$, let  $j_i = \min\{ j > j_{i-1} :  X_j = 1 \}$, i.e., $j_i$ is the index of $i$th 1 in $(X_j)$.
    We set $F_i = j_i - j_{i-1}$.
    With this coupling, $F \geq k$ implies $Y_k = \sum_{j=1}^k X_j \leq n$.
    Therefore,
    \[
    \Pr{F \geq k} \leq \Pr{Y_k \leq n},
    \]
    which we can bound using standard Chernoff bounds from Theorem~\ref{thm:chernoff}. We have that $\E{Y_k} = kp$, and
    $\mu = n\E{F_1} = n / p$.
    Then,
    \begin{align*}
        \Pr{F \geq k} &\leq \Pr{Y_k \leq n} \\
            &= \Pr{Y_k \leq \E{Y_k}\left(1 - \left(1 - \frac{\mu}{k}\right)\right)} \\
            &\leq \exp\left(-\frac{kp}{2}\left(1 - \frac{\mu}{k}\right)^2\right),
            \quad \text{by Chernoff bound,} \\
            &\leq \exp\left(-\frac{kp}{8}\right),
            \quad \text{since } k \geq 2\mu.
            \qedhere
    \end{align*}
\end{proof}


\begin{lemma}
    \label{lem:stochastic-sum}
    Let $Z_1, \dots, Z_k$ be (dependent) integer random variables, and $Z_i'$ be mutually independent random variables, such that for any $1 \leq i \leq k$ and $b \geq 0$,
    \[ \Pr{Z_i \leq b \mid Z_1, \dots, Z_{i - 1} } \geq \Pr{Z_i' \leq b}. \]
    Then, for any $b \geq 0$,
    \[ \Pr{\sum_{i=1}^k Z_i \leq b} \geq \Pr{\sum_{i=1}^k Z_i' \leq b }.\]
\end{lemma}
\begin{proof}
Follows from a simple coupling argument.
\end{proof}



\end{document}